\documentclass{article}

\usepackage{epsf} 
\usepackage{graphicx}
\usepackage{optprog}
\usepackage{multirow}
\newcounter{gap}
\newcommand{\n}{

\noindent \hspace{\thegap em}
}

\newcommand{\Procedure}{\noindent {\bf procedure }\addtocounter{gap}{2}}

\newcommand{\If}{{\bf if }}

\newcommand{\Begin}{{\bf begin}  \addtocounter{gap}{2}}

\newcommand{\Else}{{\bf else }}

\newcommand{\Do}{\addtocounter{gap}{2}{\bf do}}

\newcommand{\While}{{\bf while }}

\newcommand{\End}{{\bf end}}

\newcommand{\Comment}[1]{
{\it 
/*\\
\hspace{2em}#1\\
*/} \medskip
}

\newcommand{\EndLoop}{\addtocounter{gap}{-2}}

\newcommand{\algfigure}[4]{
\setcounter{gap}{0}
\renewcommand{\baselinestretch}{0.9}
\begin{figure}[#2]
\begin{center}
\rule{#1 \linewidth}{2.5pt}
\parbox{#1 \linewidth}{
\medskip
#3
\medskip
}
\rule{#1 \linewidth}{1.5pt}
\end{center}
\caption{#4}
\end{figure}
\renewcommand{\baselinestretch}{1.11111111111111111111111111111111}
}

\newtheorem{theorem}{Theorem}[section]
\newtheorem{lemma}{Lemma}[section]
\newtheorem{proposition}{Proposition}[section]
\newtheorem{corollary}{Corollary}[section]
\newtheorem{property}{Property}[section]
\newtheorem{definition}{Definition}[section]
\newtheorem{claim}{Claim}[section]

\newenvironment{proof}[1][Proof:]{\begin{trivlist}
\item[\hskip \labelsep {\bfseries #1}]}{\end{trivlist}}

\newcounter{contribution}
\newcommand{\qed}{\hfill \rule{2.5mm}{2.5mm}}

\begin{document}

\title{Practical and theoretical improvements for bipartite matching using the pseudoflow algorithm}
\author{Bala G. Chandran\\
\small Analytics Operations Engineering, Inc.\\
\small Boston, MA 02109. \\
\small {\tt bchandran@nltx.com}\medskip  \\
Dorit S. Hochbaum\\
\small Department of Industrial Engineering and Operations Research and\\
\small Walter A.\ Haas School of Business \\
\small University of California \\
\small Berkeley, CA 94720.\\
\small {\tt hochbaum@ieor.berkeley.edu}}

\date{}
\renewcommand{\baselinestretch}{1}

\maketitle

\renewcommand{\baselinestretch}{1.2}

\begin{abstract}
We show that the pseudoflow algorithm for maximum flow is particularly efficient for the bipartite matching problem both in theory and in practice. We develop several implementations of the pseudoflow algorithm for bipartite matching, and compare them over a wide set of benchmark instances to state-of-the-art implementations of push-relabel and augmenting path algorithms that are specifically designed to solve these problems.  The experiments show that the pseudoflow variants are in most cases faster than the other algorithms.

We also show that one particular implementation---the matching pseudoflow algorithm---is theoretically efficient.  For a graph with $n$ nodes, $m$ arcs, $n_1$ the size of the smaller set in the bipartition, and the maximum matching value $\kappa \leq n_1$, the algorithm's complexity given input in the form of adjacency lists is \mbox{$O\left (\min \{ n_1\kappa,m\} + \sqrt{\kappa}\min \{\kappa^2,m\}\right )$}.  Similar algorithmic ideas are shown to work for an adaptation of Hopcroft and Karp's bipartite matching algorithm with the same complexity.  Using boolean operations on words of size $\lambda$, the complexity of the pseudoflow algorithm is further improved to $O\left (\min \{n_1\kappa, \frac{n_1n_2}{\lambda}, m\} + \kappa^2 + \frac{\kappa^{2.5}}{\lambda}\right )$.  This run time is faster than for previous algorithms such as Cheriyan and Mehlhorn's algorithm of complexity $O\left(\frac{n^{2.5}}{\lambda}\right)$.
\end{abstract}

\section{Introduction}

The bipartite matching problem is to find, in a given bipartite graph $B=(V_1;V_2,E)$, a matching containing a maximum number of edges. That is, a collection of edges $M\subseteq E$ such that each node is adjacent to at most one of the edges in the matching $M$.  For a survey on early literature on this problem the reader is referred to the book by Lawler \cite{Law76}, Chapter 5.

The bipartite matching problem is equivalent to the maximum flow problem on an associated {\em simple} bipartite network. (A network is said to be simple if every node has a throughput capacity of 1 unit of flow.) Therefore, any maximum flow algorithm can be used to solve the bipartite matching problem. The network is constructed by adding source and sink nodes $s$ and $t$, linking the source to all nodes of $V_1$ with arcs of capacity $1$ and all nodes of $V_2$ to the sink with arcs of capacity $1$, and directing all edges in the bipartite graph from $V_1$ to $V_2$ with capacity $\geq 1$. Such a network is shown in Figure \ref{fig:simple}. The maximum $s,t$-flow on this associated network corresponds to a solution to the maximum matching problem: an edge $[i,j]: i\in V_1, j\in V_2$ is in the matching if and only if the corresponding arc $(i,j)$ has a flow of one unit on it.

\begin{figure}[ht]
\centerline{\includegraphics[width=0.4\linewidth]{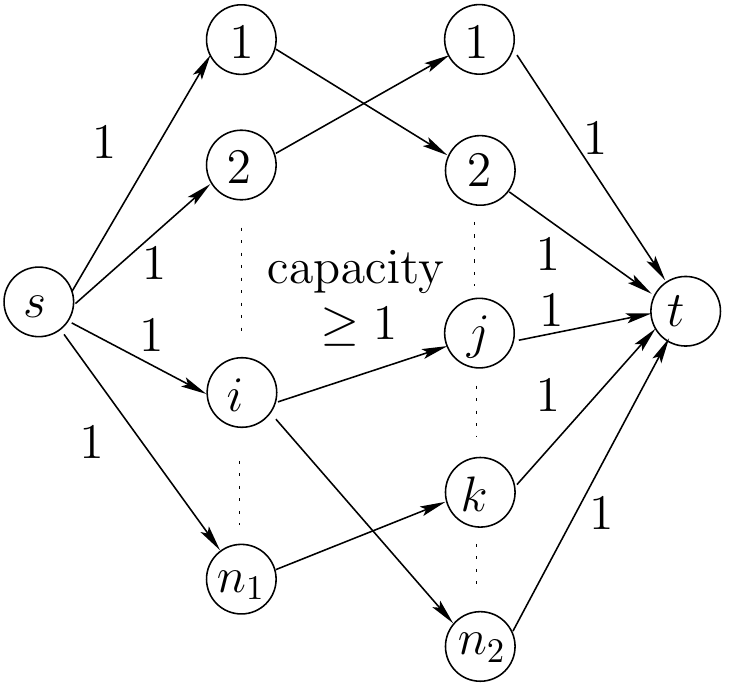}}
\caption{\label{fig:simple}Flow graph for bipartite matching.}
\end{figure}

Other than bipartite matching there are other well-known problems that are solved as maximum flow on simple bipartite networks.  These include the vertex cover problem on a bipartite graph and the independent set problem, also known as the stable set problem, on a bipartite graph. We refer to the maximum flow algorithm for simple bipartite graphs as the bipartite matching algorithm although it applies to these problems as well.

Dinic's \cite{Din70} maximum flow algorithm is particularly efficient for simple networks as demonstrated by Even and Tarjan \cite{EveT75}. For bipartite graphs, the running time is $O(\sqrt{n_1}m)$, where $n_1 = |V_1|$ (w.l.o.g. $n_1 \leq n_2=|V_2|$).  Hopcroft and Karp \cite{HopK73} proposed an algorithm for bipartite matching with complexity $O(\sqrt{\kappa}m)$, where $\kappa$ is the cardinality of the maximum matching which is bounded by $n_1$. Their algorithm is, in essence, the same as Dinic's algorithm adapted to bipartite matching. Feder and Motwani \cite{FedM91} obtained a bound of $O(\sqrt{n}m^{*})$ for the bipartite matching algorithm that relies on speeding up Dinic's algorithm using graph compression. In the complexity expression, $m^{*}$ is the number of edges in the compressed graph, which is less than $m$ by about a factor of $\log n$.  Using boolean word operations on $\lambda$-bit words, Cheriyan and Melhorn \cite{CheM96} obtained a bound of $O(\frac{n^{2.5}}{\lambda})$ while Alt et al. \cite{AltBMP91} obtained a bound of $O(n^{1.5}\sqrt{\frac{m}{\lambda}})$ (which is better than $O(\frac{n^{2.5}}{\lambda})$ for sparse graphs and better than $O(\sqrt{n}m)$ for dense graphs).  Mucha and Sankowski \cite{MucS04} described a randomized algorithm for matching in general (non-bipartite) graphs that runs in $O(n^\omega)$, where $\omega=2.38$ is the exponent of the best known matrix multiplication algorithm.

However, the theoretically efficient algorithms listed above tend to perform poorly in practice. Setubal \cite{Set93, Set96} showed that in practice, implementations of the push-relabel algorithm of Goldberg and Tarjan \cite{GolT88} were faster than those of Dinic's as well as the algorithm of Alt et al.  Cherkassky et al. \cite{CheGMSS98} developed several implementations of push-relabel and performed extensive experiments on several benchmark instances, showing push-relabel to be the fastest in practice.

In this paper, we apply the pseudoflow algorithm of Hochbaum \cite{Hoc97, Hoc07} to bipartite matching and examine its theoretical and practical performance.  The pseudoflow algorithm was recently shown by Chandran and Hochbaum \cite{ChaH07} to be the fastest algorithm in practice for the maximum flow problem, and by Hochbaum and Orlin to be as efficient as the push-relabel algorithm in theory; hence, it is reasonable to suspect that the pseudoflow algorithm is efficient for bipartite matching as well.  The major contributions of our work are as follows.
\begin{enumerate}
\item We develop several implementations of the pseudoflow algorithm specifically for bipartite matching and show that are faster than state-of-the-art implementations of push-relabel for bipartite matching.  We use the results of the experiments to gain insights into the differences between the pseudoflow and push-relabel algorithms. 
\item We show that a variant of the pseudoflow algorithm, called the {\sf matching-pseudoflow} algorithm, runs on a bipartite simple network in time $O(\min\{n_1\kappa, m\}+\sqrt{\kappa}\min\{\kappa^2, m\})$. We then show that the insights generated from this approach allow to modify either Hopcroft and Karp's algorithm or the push-relabel maximum flow algorithm and achieve the same complexity.  Using boolean operations on $\lambda$-bit words, we show that the complexity of the {\sf matching-pseudoflow} algorithm can be further improved to $O\left (\min \left \{ m,n_1\kappa, \frac{n_1n_2}{\lambda} \right \} + \kappa^2 +\frac{\kappa^{2.5}} {\lambda}\right )$.

Since the {\sf matching-pseudoflow} algorithm could be viewed as a superior implementation of Dinic's algorithm, we compare the performance of the {\sf matching-pseudoflow} to the best-known implementation of Dinic's algorithm to understand and quantify the key differences between the two algorithms. 
\end{enumerate}

\section{Description of the pseudoflow algorithm}
\label{sec:pseudoflow}

The pseudoflow algorithm and its properties are described in detail in
Hochbaum \cite{Hoc07}.  The description is repeated here for
completeness.

\subsection{Preliminaries}
\label{section:prelims}

Let $G_{st}$ be a graph $(V\cup\{s,t\}, A\cup A_s \cup A_t)$, where $A_s$ and $A_t$ are the source-adjacent and sink-adjacent arcs respectively.

A flow vector $f = \{f_{ij}\}_{(i,j) \in A\cup A_s \cup A_t}$ is said to be {\em feasible} if it satisfies
\begin{enumerate}
\item Flow balance constraints: for each $i \in V$, $\sum_{(k,i)\in A\cup A_s \cup A_t} f_{ki} = \sum_{(i,j)\in A\cup A_s \cup A_t} f_{ij}$ (i.e., inflow($i$) = outflow($i$)), and
\item Capacity constraints: the flow value is between the lower bound and upper bound capacity of the arc, i.e.,  $\ell_{ij} \leq f_{ij} \leq u_{ij}$.  Without loss of generality, we assume henceforth that $\ell_{ij} = 0$ (e.g., Ahuja at al. \cite{AhuMO93}, pages 191--196).
\end{enumerate}

A {\em maximum flow} is a feasible flow $f^*$ that maximizes the flow out of the source (or into the sink).  The value of the maximum flow is $\sum_{(s,i) \in A_s} f^*_{si}$.

Given a flow vector $f$ in $G_{st}$ that is feasible, the {\em residual graph} $G^f=(V \cup \{s,t\}, A^f)$ is constructed as follows: for each arc $(i,j) \in A \cup A_s \cup A_t$ with flow $f_{ij}$ and capacity $c_{ij}$, $A^f$ contains two arcs: $(i,j)$ with capacity $c_{ij} - f_{ij}$ and $(j,i)$ with capacity $f_{ij}$.  The capacities of arcs in $A^f$ are referred to as the residual capacities with respect to flow $f$, and are denoted by $c^f$.  An {\em $s,t$-cut} in the graph is a bi-partition of nodes into two disjoint sets -- one containing the source and the other containing the sink.  One property of the residual graph is that the bipartition of nodes of the minimum $s,t$-cut of $G^f$ is the same as that in $G$ (e.g., Ahuja at al. \cite{AhuMO93}, pages 44--46).

A {\em pseudoflow} $f$ is a flow vector that satisfies capacity constraints, but may violate flow balance at any node. The {\em excess} of a node $v\in V$ is the inflow into that node minus the outflow denoted by $e(v) = \sum_{(u,v)\in A\cup A_s\cup A_t} f_{uv} - \sum_{(v,w)\in A\cup A_s\cup A_t} f_{vw}$. A negative excess is called a {\em deficit}.

A tree $T = (V, E)$ is a connected, undirected, acyclic graph.  A rooted tree has a distinguished node $w$ called the root. For each edge $[u,v]$, $u$ is said to be the {\em parent} of $v$ if $u$ is closer to the root than $v$, and is denoted by ${\rm parent}(v)$. Node $v$ is then called the {\em child} of $u$, and is denoted by child($u$). The only node in the tree that does not have a parent is the root. A node $v$ is said to be an {\em ancestor} of a node $u$ if $v$ lies along the unique path from $v$ to the root; node $u$ is then said to be a {\em descendant} of node $v$. For convenience, we will assume that the tree points topologically ``downward'' with the root at the ``top'' of the tree, and each node ``below'' its ancestors. A {\em branch} rooted at some node $r$ is a sub-graph of the tree that contains $r$ and all its descendants in the tree.  A rooted {\em sub-tree} is a connected sub-graph of the given tree (unlike a branch, it need not contain all the descendants of its root). 

An arc that carries a flow equal to its upper bound is said to be {\em saturated}.  The pseudoflow algorithm maintains a flow that saturates source-adjacent and sink-adjacent arcs throughout the algorithm. Consequently, the source and sink have no further role in the algorithm and are contracted into a single node $r$ that ``keeps track'' of the excesses and deficits of the nodes in $V$ by adding excess and deficit arcs as follows: For each node $v \in V$ with positive excess, we add to the graph an arc $(v,r)$ called an {\em excess arc}, and for each node $u \in V$ with negative excess we add an arc $(r,u)$ called a {\em deficit arc}. The network thus obtained is referred to as the extended network $G^{ext} = (V \cup \{r\}, A \cup A_r)$, where $A_r$ is the set of excess and deficit arcs.

For a tree $T$, an arc $(u,v)$ is said to be {\em in-tree} if the edge $[u,v]\in T$. Arcs that are not in tree are said to be {\em out-of-tree}.  Given a pseudoflow $f$ that saturates $A_s$ and $A_t$, a {\em normalized tree} is a tree in $G^{ext}$ rooted at $r$ that satisfies the following three properties. 
\begin{property}
\label{property:normtreeroot}
The nodes that do not satisfy flow balance constraints are the children of $r$ and are the roots of their respective branches.
\end{property}
\begin{property}
\label{property:normtreeoutoftree}
The pseudoflow values of $f$ on out-of-tree arcs are at the lower or upper bound capacities of the respective arcs.
\end{property}
\begin{property}
\label{property:normtreedownward}
In every branch, all downward residual capacities are strictly positive.
\end{property}
A schematic description of a normalized tree is shown in Figure \ref{Figure:normtree}.

\begin{figure}
\centerline{\includegraphics[width=0.7\linewidth]{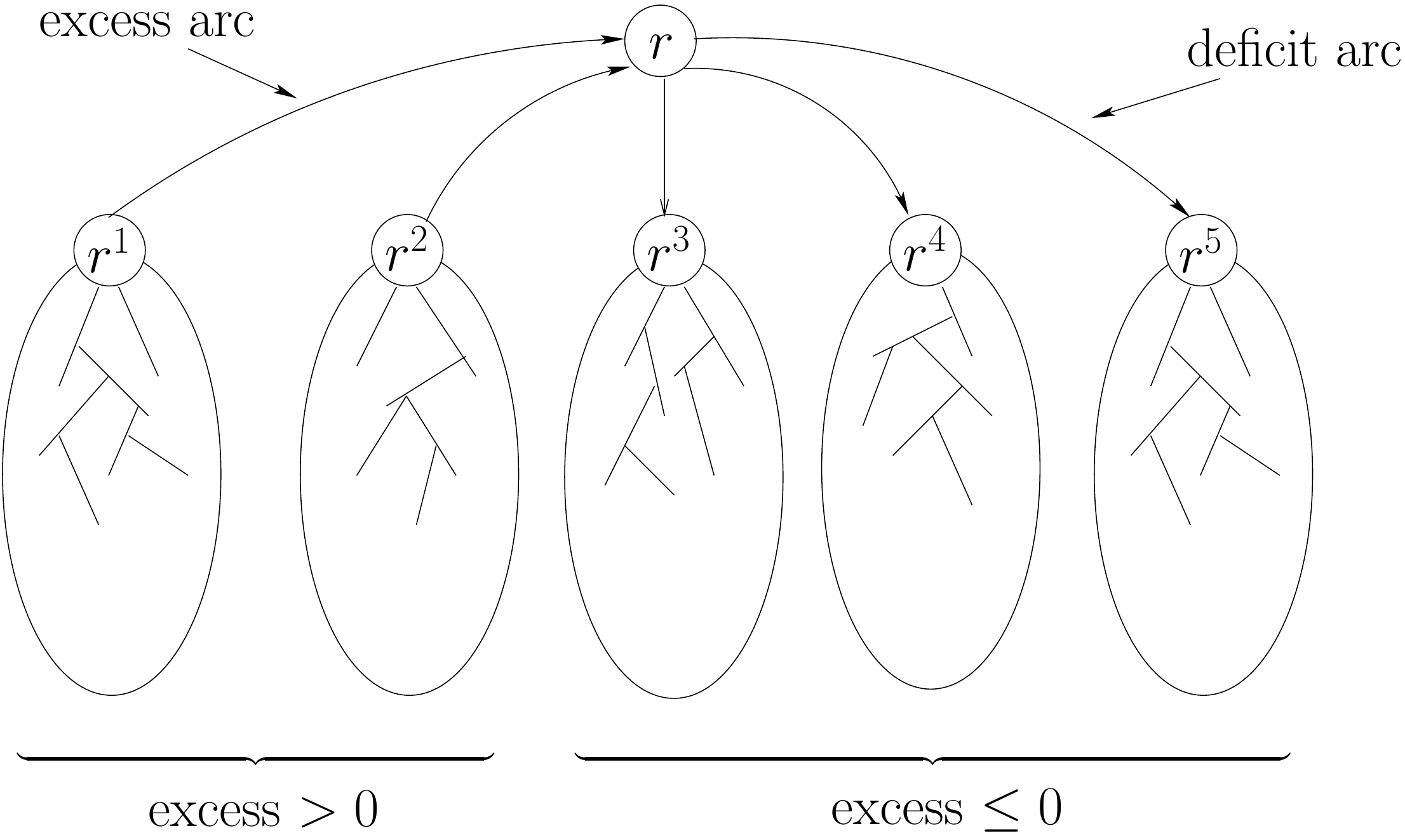}}
\caption{\label{Figure:normtree}A schematic description of a normalized tree. Each $r^i$ is the root of a branch.}
\end{figure}

The pseudoflow algorithm starts with any normalized tree and an associated pseudoflow.  The generic initialization is the {\em simple} initialization: source-adjacent and sink-adjacent arcs are saturated while all other arcs have zero flow.

If a node $v$ is both source-adjacent and sink-adjacent, then at least one of the arcs $(s,v)$ or $(v,t)$ can be pre-processed out of the graph by sending a flow of $\min\{c_{sv}, c_{vt}\}$ along the path $s\rightarrow v \rightarrow t$. This flow eliminates at least one of the arcs $(s,v)$ and $(v,t)$ in the residual graph. We henceforth assume w.l.o.g. that no node is both source-adjacent and sink-adjacent.

The simple initialization creates a set of source-adjacent nodes with excess, and a set of sink-adjacent nodes with deficit.  Since all other arcs have zero flow, they are all out-of-tree arcs. Thus, each node is a singleton branch for which it serves as the root, even if it is {\em balanced} (with $0$-deficit). The simple initialization results in a simple normalized tree shown in Figure \ref{Figure:simpletree}.

\begin{figure}
\centerline{\includegraphics[width=0.6\linewidth]{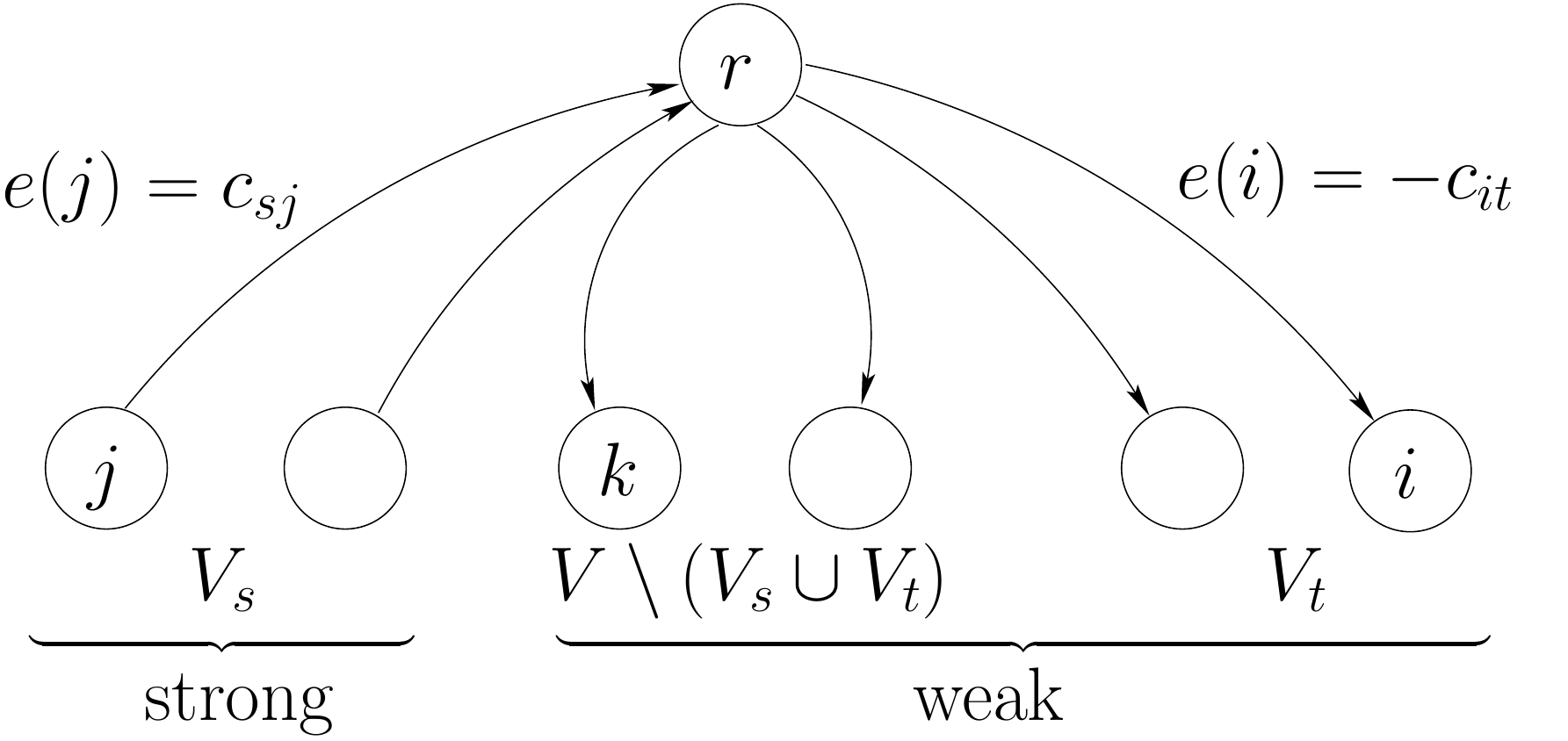}}
\caption{\label{Figure:simpletree}A simple normalized tree.}
\end{figure}

\subsection{A labeling pseudoflow algorithm}
\label{Section:labelingalgo}

In the labeling pseudoflow algorithm, all nodes carry a label $\ell _v$ for all $v\in V$. Initially, all labels are set to the value $1$. An iteration of the algorithm consists of identifying a branch with root carrying strictly positive excess, and attempting to push this excess towards the sink through the residual network. The process of pushing excesses towards the sink is performed via a {\em merger}. Given a branch with root of label $\ell$ and positive excess, a merger operation consists of identifying a {\em merger arc} with positive residual capacity from a node of label $\ell$ within the branch to some node of label $\ell-1$ in the graph.

A {\em relabeling} of a node is the increase of a node's label by one unit.  A node of label $\ell$ is relabeled to $\ell +1$ if there is no merger arc in the residual graph to a neighbor of label $\ell-1$, and if all its children in the branch have label at least $\ell +1$. With these rules, the labels satisfy the following properties.

\begin{lemma}[Hochbaum \cite{Hoc97, Hoc07}]
\label{lem:labels}
For the labeling pseudoflow algorithm, the labels satisfy:\\
(a) For every residual arc $(u,v)$, $\ell _u\leq \ell _v +1$.\\
(b) The labels of nodes are monotone nondecreasing in the downwards direction in each branch.
\end{lemma}

\begin{corollary}[Hochbaum \cite{Hoc97, Hoc07}]
\label{cor:res-path} The label assigned to a node throughout the labeling pseudoflow algorithm does not exceed the length of a shortest path to a sink-adjacent node in the residual graph plus the label of the sink-adjacent node.  More generally, the positive difference in labels of two nodes does not exceed the length of the residual path between them.
\end{corollary}

For convenience, we henceforth refer to the branch containing the tail of the merger arc as the ``from-branch'' and the one containing the head of the merger arc as the``to-branch''. Once a merger arc is identified, a merger operation is performed on the normalized tree. This consists of adding the merger arc to the normalized tree, and removing the arc from the root of the from-branch to the root of the normalized tree. The merger operation is shown in Figures \ref{Figure:beforeafter}(a) and (b). At the end of a merger, the tree is not a normalized tree since it has a non-root node carrying positive excess. The merged branch is now {\em renormalized}, a process that may create any number of branches out of the merged branch. The process of renormalization of the merged branch consists of pushing the excess of the root of the from-branch towards the root of the to-branch and updating the pseudoflows and excesses. The path from the root of the from-branch to that of the to-branch is unique since they are nodes in a connected tree. For each edge on this path, the operation of pushing the excess from the child to its parent and updating the pseudoflow on the edge is called a {\em push}. If only a part of the child's excess can be pushed to its parent due to insufficient residual capacity on that arc, the child retains some positive excess. The edge to its parent is then removed from the normalized tree and an excess arc is added for the child node making it the root (with positive excess) of a branch consisting of all nodes below it. This operation, called a {\em split}, is shown in Figures \ref{Figure:beforeafter}(b) and \ref{Figure:beforeafter}(c).

\begin{figure}
\centerline{\includegraphics[width=\linewidth]{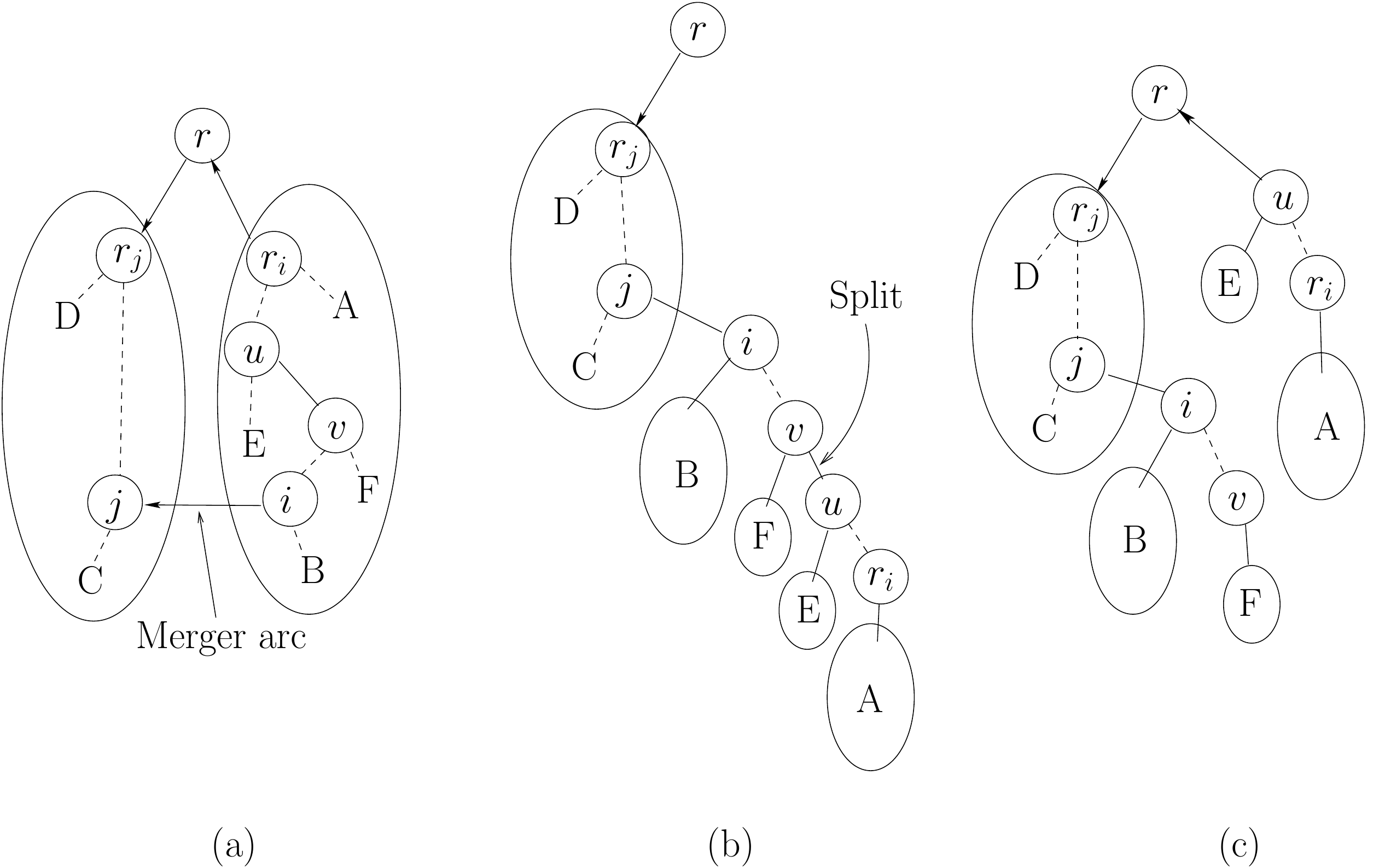}}
\vspace{0.1in}
\caption{\label{Figure:beforeafter} (a) Initial normalized tree, (b) Tree obtained after the merger, (c) Re-normalized tree after split due to insufficient residual capacity on edge $[v,u]$.}
\end{figure}

If the root of a branch is relabeled to label $n$ at some point in the algorithm, all nodes in this branch have label $n$. By Corollary \ref{cor:res-path}, this implies that all deficit nodes are unreachable from nodes of label $n$. Hence, all nodes in the branch must be in the source set of a minimum cut, and can be ignored for the remainder of the algorithm. Thus, the algorithm terminates when (i)~there are no branches with root carrying positive excess, or (ii)~all such roots have a label of $n$.

When the algorithm terminates, we obtain a normalized tree and a pseudoflow where all nodes belonging to branches with positive excess (if they exist) have label $n$. This is not a feasible flow since the normalized tree has excess and deficits. However, the normalized tree contains information regarding a minimum cut, which is stated in the following theorem.

\begin{theorem}[Hochbaum \cite{Hoc97, Hoc07}]
The source node along with all nodes of label $n$ in the normalized tree form the source set of a minimum cut while the remaining nodes form the sink set of a minimum cut.
\end{theorem}

\subsection{Implementation details}

\noindent{\bf Limiting the number of arc scans:~} During the labeling algorithm, the arcs adjacent to each node are examined at most once (see Hochbaum \cite{Hoc97, Hoc07}) for each value of the node's label.  To implement this, we maintain a pointer at each node to the arc that was last scanned to find a merger. If any node is visited more than once for a given label, the search for mergers resumes from the last scanned arc, thus ensuring that each arc is scanned at most once for each label. When a node is relabeled, the pointer is reset to the start of its list of adjacent arcs.

\noindent{\bf Root management:~} The labeling algorithm requires that all roots with positive excess and of a particular label be available when queried. To achieve this, the roots are maintained in an array of buckets, where a bucket contains all roots with positive excess and with a particular label. The order in which roots within a bucket are processed for mergers appears to make a difference to the pseudoflow algorithm. Anderson and Hochbaum \cite{AndH02} experimented with three branch management policies:
\begin{itemize}
\item {\bf FIFO:~} Each bucket is maintained as a queue; roots are added to the rear of the queue, and roots are retrieved from the front of the queue.
\item {\bf FIFO:~} Each bucket is maintained as a stack; roots are added to the top of the stack, and roots are retrieved from the top of the stack.
\item {\bf Wave:~} This is a variant of the LIFO policy.  Each bucket is still maintained as a stack, with roots being added to the top of the stack and being retrieved from the top.  However, when the excess of a root changes while it is in the bucket, it is moved up to the top of the stack.

Note that the wave management policy is the same as the LIFO policy for the lowest label variant since the excess of a root with positive excess does not change while it is in a bucket. (When a root is processed in the lowest label algorithm, all mergers are from a branch with positive excess to one with non-positive excess, leaving all other roots with positive excess unchanged.)
\end{itemize}

\noindent{\bf Gap Relabeling:~} We use the gap-relabeling heuristic of Derigs and Meier \cite{DerM89}, who introduced it in the context of push-relabel. When we process a branch whose root has label $\ell$ and there are no nodes in the graph with label $\ell-1$, we conclude that the entire branch has no residual paths to the sink and is hence a part of the source set of a min cut. The entire branch can thus be ignored for the rest of the algorithm. In practice, this is achieved by setting the labels of all nodes in that branch to $n$.

The {\em Min-cut Stage} refers to all the operations executed until a minimum cut is obtained.

\subsection{Lowest and highest label pseudoflow variants}

In the generic labeling algorithm, the branch with a root carrying positive excess that is selected for processing (finding mergers) is chosen arbitrarily. In the {\em lowest label variant}, the root carrying positive excess with the lowest label is identified and the branch is processed for mergers so long as its root remains the lowest labeled root with positive excess.  In the {\em highest label variant}, the branch that is chosen is the one with root of highest label, i.e., at each iteration the root carrying positive excess with highest label is identified and that branch is processed.  Note that in the lowest label variant, the root of the from-branch has positive excess while that of the to-branch has non-positive excess, while in the highest label variant, roots of both the from-branch and to-branch could have positive excess.

\section{Complexity of the pseudoflow algorithm for bipartite matching}
\label{section:bipcomplexity}

We now analyze the complexity of the highest and lowest label psuedoflow algorithms when applied to bipartite matching.

\begin{definition}
The algorithm is said to be in {\em phase $\ell$} when nodes of label $\ell$ are being examined for mergers.
\end{definition}

Let the cardinality of the maximum matching in $G$ be $\kappa$.  Since the graph is bipartite, every alternate node in any path in the network must be a $V_2$-node.  The shortest path from any node in the network to a node with strict deficit (i.e., an unmatched $V_2$ node) can contain at most $\kappa$ matched $V_2$-nodes, hence its length is at most $2 \kappa$.  By Corollary \ref{cor:res-path} this means that the label of each node (and hence the number of phases) for the lowest label algorithm is $O(\kappa)$, while that for the highest label algorithm is $O(n_1)$.

\begin{proposition}
The depth of the normalized tree is $O(\kappa)$.
\end{proposition}
\begin{proof}
Consider a path from a node up to the root of the normalized tree.  Since the graph is bipartite, the path is made up of alternating nodes from $V_1$ and $V_2$.  Thus, every alternate edge in the path is a valid matching, which bounds the length of the path (and thus the depth of the tree) by $2 \kappa$.\qed
\end{proof}

The implication of the above proposition is that the work done per merger is $O(\kappa)$.

\begin{proposition}
The number of arc scans in the lowest pseudoflow algorithm for bipartite matching is $O(\min\{\kappa m, n_1^2 \kappa\})$.
\end{proposition}
\begin{proof}
Hochbaum \cite{Hoc97, Hoc07} showed that each arc is examined $O(1)$ times per phase.  Since there are $O(\kappa)$ phases and $m$ arcs, the total number of arc scans is $O(\kappa m)$.

Each time a node is processed, its neighbors are examined in order to find a merger.  Since there are $O(n_1)$ nodes in the normalized tree, at most $n_1$ neighbors need to be examined in order to find a merger or determine that no merger exists.  Thus, the total number of arc scans is the number of nodes in the normalized tree times the number of arc scans per phase times the number of phases, which is $O(n_1^2 \kappa)$. \qed
\end{proof}

Similarly, for the highest label algorithm, the number of arc scans $O(\min\{n_1 m, n_1^3\})$.

Following the pseudopolynomial complexity analysis of the generic lowest label pseudoflow algorithm from Hochbaum \cite{Hoc97, Hoc07}, we get a bound of $O(n_1 \kappa)$ on the number of mergers for the lowest label variant.  The number of mergers in the highest label pseudoflow algorithm is $O(n_1 m)$ (as shown by Hochbaum \cite{Hoc97, Hoc07}, the number of mergers is bounded by $m$ times the number of phases).

The total work done in the pseudoflow algorithm is the number of arc scans plus the number of mergers times work per merger (which is $O(\kappa)$ as shown above).  Thus, the complexity of the lowest label pseudoflow algorithm for bipartite matching is \mbox{$O(\min\{\kappa m, n_1^2 \kappa\} + n_1 \kappa^2)$}, while that of the highest label algorithm is $O(\kappa n_1 m)$.

\section{The free-arcs pseudoflow algorithm for bipartite matching}

In the {\em free-arcs} version of the pseudoflow algorithm, the normalized tree satisfies the following property in addition to Properties \ref{property:normtreeroot} through \ref{property:normtreedownward}.

\begin{property}
\label{property:normtreeupward}
In every branch, all upward residual capacities are strictly positive.
\end{property}

The only difference from the perviously described pseudoflow algorithm is in the {\sf split} operation, which is now initiated if the upward residual capacity of an in-tree arc becomes zero after a push.

The implication of the above property is that the normalized tree contains only ``free'' arcs, i.e., arcs that have flow strictly between their lower and upped bounds.

Given a bipartite graph $G = (V_1; V_2, E)$, the flow network is constructed by adding source and sink nodes $s$ and $t$, linking the source to all nodes of $V_1$ with arcs of capacity $1$ and all nodes of $V_2$ to the sink with arcs of capacity $1$, and directing all edges in the bipartite graph from $V_1$ to $V_2$ with {\em infinite} capacity.  For the free-arcs algorithm, the infinite capacity on arcs from $V_1$ to $V_2$ implies that all in-tree arcs have unit flow while all out-of-tree arcs have flow equal to the lower bound of zero (the flow on an arc can never be at its upper bound).

\begin{lemma}
\label{lem:branchtypes} The pseudoflow algorithm for bipartite matching can create only four types of branches -- two types of strong branches $ST_1$ and $ST_2$, and two types of weak branches $WT_1$ and $WT_2$ (as described in Figure \ref{fig:branchtypes}(a)).
\end{lemma}

\begin{figure}[ht]
\centerline{\includegraphics[width=0.9\linewidth]{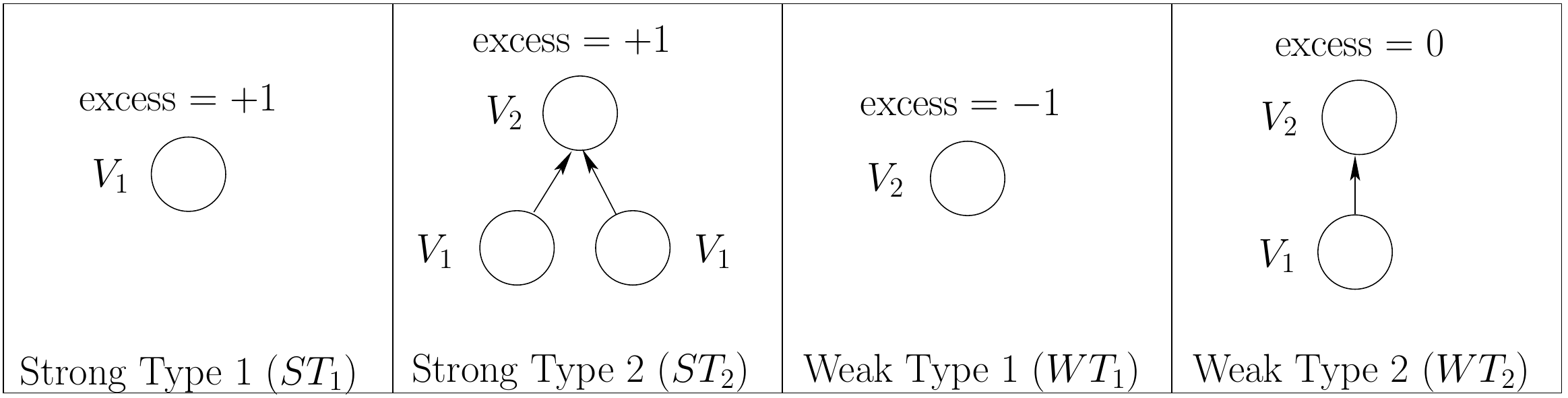}}
\caption{\label{fig:branchtypes}Types of branches that can exist in a normalized tree during execution of the free-arcs pseudoflow algorithm for bipartite matching.}
\end{figure}

\begin{figure}[ht]
\centerline{\includegraphics[width=\linewidth]{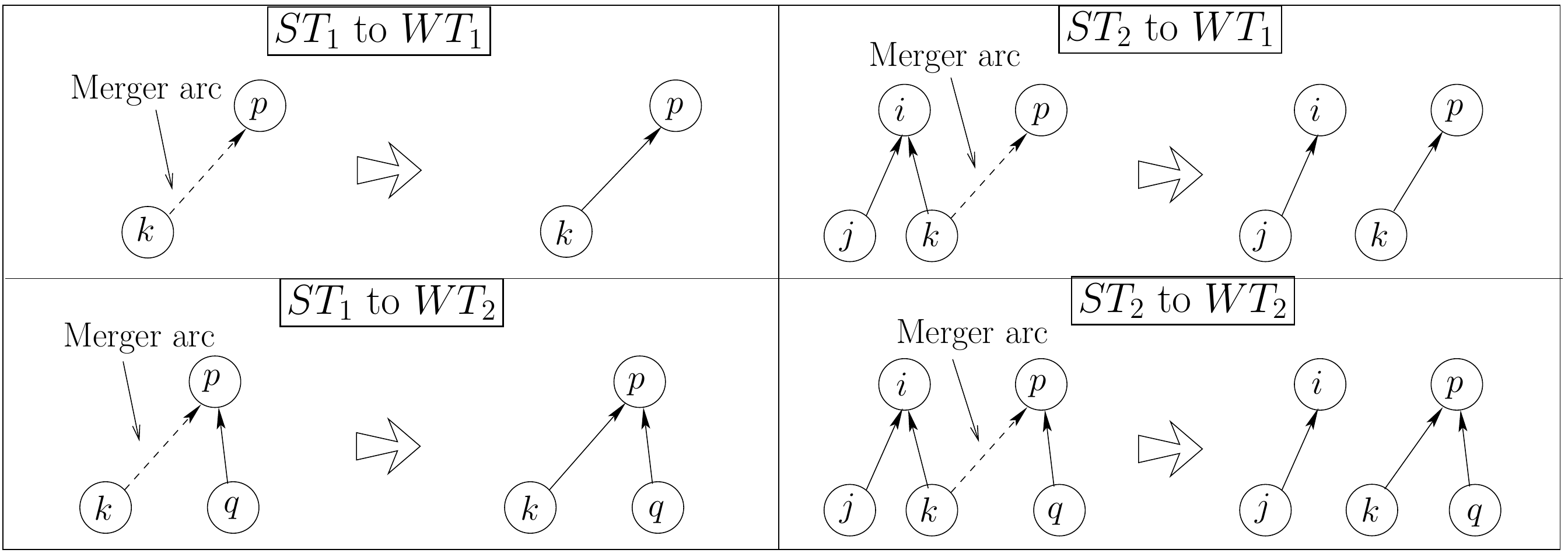}}
\caption{\label{fig:mergertypes}Types of mergers that can occur in the free-arcs pseudoflow algorithm for bipartite matching.}
\end{figure}

\begin{proof}
The proof is by induction.  The inductive assumption applies initially as in the simple normalized tree all nodes of $V_1$ are $ST_1$ branches and all nodes of $V_2$ are $WT_1$ branches.  Given that an iteration starts with these two types of strong branches and two types of weak branches, only four types of mergers are possible as shown in Figure \ref{fig:branchtypes}(b).  All these mergers result in one or two of these types of branches, and thus the proof is complete. \qed
\end{proof}

Each $WT_2$ and $ST_2$ branch contains an edge between a $V_1$ node and a $V_2$ node, and all branches are node-disjoint. Thus, the set of $WT_2$ and $ST_2$ branches represent a valid matching, which leads to the following property.

\begin{property}
\label{property:branchbound}
The number of $ST_2$ and $WT_2$ branches is bounded by $\kappa$, the cardinality of the maximum matching.
\end{property}

\subsection{Complexity of the free-arcs pseudoflow algorithm for bipartite matching}

All the results in Section \ref{section:bipcomplexity} are still valid, except that the work done per merger is now $O(1)$.  The complexity of the free-arcs version of the pseudoflow algorithm is thus the number of arc scans plus the number of mergers, which is \mbox{$O(\min\{\kappa m, n_1^2 \kappa\})$} for the lowest label variant and $O(n_1 m)$ for the highest label variant.

\section{The matching-pseudoflow algorithm}
\label{section:matching-pseudoflow}

The {\sf matching-pseudoflow} algorithm is a pseudoflow algorithm with global relabeling and delayed relabeling.  We first introduce the concept of two-edge distance labels in the graph.

\begin{definition}
The {\em two-edge distance label} of a node is the number of $V_1$-nodes in the shortest path from that node to a $WT_1$ branch in the residual network.
\end{definition}
The notion of two-edge distances in bipartite graphs has been used previously e.g.\ by Ahuja et al. \cite{AhuOST94}.  Initially, all labels of nodes in $V_1$, which are $ST_1$ branches, are set to $1$ and the labels of nodes in $V_2$, which are $WT_1$ branches, are set to $0$. Throughout the algorithm the labels of nodes that form $WT_1$ branches remain $0$ as their two-edge distance (to themselves) is $0$.

{\em Delayed relabeling} means that all possible mergers from lowest labeled strong nodes of label $\ell$ are performed
{\em without relabeling the nodes} when no merger is found.  Once all the nodes of label $\ell$ have been examined for mergers, all the node labels are set to be the shortest two-edge distance to a $WT_1$ branch in the residual graph and the set of all mergers starting from the lowest labeled strong root are again performed.  The process of computing all
the node distance labels is referred to as {\em global relabeling} \cite{GolC97}.

We now demonstrate that the two-edge distance labels satisfy properties analogous to (a) and (b) of Lemma \ref{lem:labels}. Property (a) is satisfied by the distance labels of nodes which are the lengths of the shortest residual path from each node to the sink.  The second property of monotonicity (b) is shown to be satisfied next.

\begin{lemma}
\label{lem:monotone}
The two-edge distance labels satisfy property (b) in Lemma \ref{lem:labels}.
\end{lemma}

\begin{proof}
Two-edge distance labels satisfy that {\em both nodes in a $WT_2$ branch have the same label}: In a $WT_2$ branch, all arcs into its root (a $V_2$-node) other than that from its child (a $V_1$-node) carry zero flow; the arc from its child carries a flow of 1 unit.  The arc from its child is the only arc with positive residual capacity adjacent to the root. Thus, the root can reach a $WT_1$ node only through its child and the shortest path from the root to a $WT_1$ branch will contain the shortest path from its child to a $WT_1$ branch.  So the number of $V_1$-nodes in the shortest path from the root to a $WT_1$ branch will be the same as that in the shortest path from its child to a $WT_1$ branch, ensuring that the root and child have the same two-edge distance label.

A similar argument holds for the $ST_2$ branches. Let $\ell_R$ be the label of the right child and $\ell_L$ be the label of the left child (assume w.l.o.g. that $\ell_R \leq \ell_L$). The root of an $ST_2$ branch can reach a $WT_1$ branch only through on of its children; so the label of the root from a $WT_1$ branch will be equal to $\ell_R$, the smaller label of the two children. Since there is a residual arc (of infinite capacity) from the left child to the root, the two-edge distance from the left child to the right is 1. Hence, $\ell_L \leq \ell_R+1$, and the label of the left child is at least equal to and at most one greater than the label of the root. \qed
\end{proof}

\begin{definition}
{\em Stage $\ell$} of the algorithm is the maximal set of mergers that occur while the shortest two-edge distance from a strong node to a $WT_1$ branch is $\ell$.
\end{definition}

An initialization procedure, equivalent to a stage $1$, creates a maximal set of $WT_2$ branches by scanning the neighbors of each $V_1$-node to identify an unmatched $V_2$-node and then performing a merger. Since the cardinality of the maximum matching is at most $\kappa$, we need to scan at most $\kappa$ neighbors of each $V_1$-node to identify an unmatched $V_2$-node or determine that none exists. Also, each arc is scanned at most once, so the complexity of this procedure is $O(\min\{m, n_1 \kappa\})$. At the end of the initialization, the shortest two-edge distance from a $ST_1$ branch to a $WT_1$ branch is at least 2.

We now elaborate on the implementation of a stage.  To facilitate the description, we introduce the following notation for labels of nodes in a branch. The labels of an $ST_2$ branch are represented by the triplet $(left,root,right)$ which represent the labels of the left child, root, and right child respectively. We will assume w.l.o.g. that the left child has label greater than or equal to that of the right child.  Labels in a $WT_2$ branch are represented by the pair $(child,parent)$ which represent the labels of the child and parent respectively.

At the beginning of each stage, global relabeling is performed, and all nodes are ``unflagged'', which marks them as being {\em unvisited}. Mergers are allowed only between unvisited nodes.

The merger/split operations at each stage are such that they satisfy the property that the stage begins and ends with only $WT_1$, $ST_1$, and $WT_2$ branches; $ST_2$ branches are only formed temporarily during a stage. This inductive property holds initially for stage 2 since no $ST_2$ branches are formed in stage 1 (the greedy initialization).

Suppose that at the beginning of stage $\ell \geq 2$, the set of branches consists of $ST_1$ branches of label $\geq \ell$; $WT_2$ branches in which both the nodes have the same label $p$ ($1 \leq p < \ell$); and $WT_1$ branches which have label $0$. Consider a sequence of mergers starting from a $ST_1$ branch of lowest label $\ell$.  The first merger is from a $ST_1$ branch of label $\ell$ to an unvisited root of a $WT_2$ branch with label $(\ell-1,\ell-1)$. This creates a $ST_2$ branch $(\ell, \ell-1, \ell-1)$. This branch now has the lowest labeled strong root, and the search for mergers starts from the right child labeled $\ell-1$.

Suppose that at some point a merger results in a $ST_2$ branch $(p+1,p, p)$. The search for mergers now starts from the right child of this branch resulting in one of the following possible outcomes.
\begin{enumerate}
\item There is no merger to an unvisited weak node of label $p-1$: Here we {\em delay} the relabeling of that node to the end of the stage and mark the root of the branch as being {\em visited} implying that the branch cannot participate in any more mergers at the current stage.  In this case, a {\em backtrack} operation is performed to reverse the last merger and restore the structure of the branches to what it was prior to the last merger. This is shown in Figure \ref{fig:dfstree2}. For example, consider the case where the backtrack operation occurs from a branch of label $(\ell,\ell-1, \ell-1)$ in stage $\ell$.  Suppose there are no mergers from the right child of this branch, then the backtrack operation splits the branch, creating a $WT_2$  branch of label $(\ell-1, \ell-1)$ and one $ST_1$ branch of label $\ell$. The root of the $WT_2$ branch is marked as visited, and the search for mergers continues from the $ST_1$ branch. If no more mergers are possible from this node, it is marked as visited, and a new lowest labeled strong node is picked.  This procedure continues until there are no more unvisited $ST_1$ nodes of label $\ell$.
\item $p > 1$ and a merger is found to an unvisited root of a $WT_2$ branch of label $(p-1,p-1)$: This creates a $ST_2$ branch of label $(p, p-1, p-1)$ and the search for mergers continues from the right child of this branch.
\item $p=1$ and a merger is found to a $WT_1$ branch which has label $0$: This creates a new $WT_2$ branch $(1,0)$, incrementing the size of the current matching. The branches involved in this sequence of mergers are all marked as visited, and do not participate in any more mergers in stage $\ell$. The process of searching for mergers then starts with an unvisited $\ell$ labeled strong node if there is one, or else the stage terminates.
\end{enumerate}

\begin{figure}[ht]
\centerline{\includegraphics[width=\linewidth]{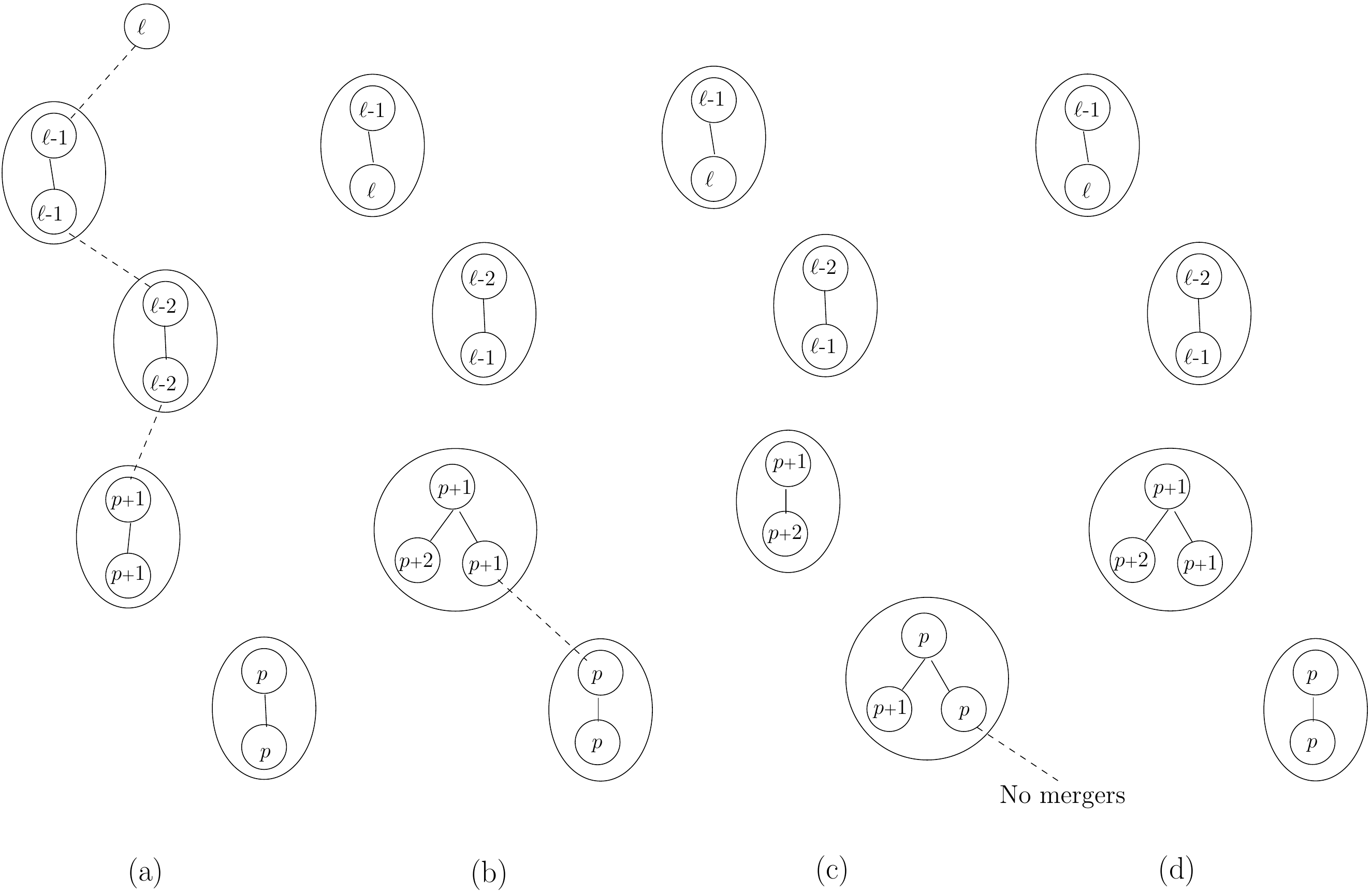}}
\caption{\label{fig:dfstree2}
(a) Branches before mergers (shown in dotted lines),
(b) Branches after mergers until label $p+1$,
(c) Branches after merger from $p+1$ to $p$,
(d) Branches after lack of merger causes a backtrack.
The branch with label $p$ is marked as having been visited.}
\end{figure}
\noindent We have thus proved the following lemma, which holds inductively given that stage 1 ends with $ST_1$, $WT_1$, and $WT_2$ branches.

\begin{lemma}
\label{property:threeTypes}
With the merger/split/backtrack operations described above, a stage that begins with only $WT_1$, $WT_2$, and $ST_1$ branches terminates with only these three branch types.
\end{lemma}

\begin{definition}
A {\em successful} path of length $\ell$ is a sequence of mergers, at stage $\ell$, that starts at a $ST_1$ branch of label $\ell$ and ends at a $WT_1$ branch of label $0$.
\end{definition}

A successful path  contributes to the increase of the number of $WT_2$ branches by $1$ which is equivalent to increasing the size of the matching. The mergers that form a successful path are called {\em successful mergers}. The next lemma proves that the procedure of flagging nodes as visited does not block off any successful paths, implying that {\em all} successful paths of length $\ell$ are found in stage $\ell$.

\begin{lemma}
\label{lem:visited} A node that is marked as visited in stage $\ell$ can no longer be part of a successful path of length $\ell$.
\end{lemma}
\begin{proof}
If a $V_2$-node has been marked as visited after no merger have been found from its right child, then that child cannot lead to any merger in the current stage.  Hence, once a backtrack occurs, the $V_2$-node is marked as visited, and need not be visited again during that stage.

If an unvisited $V_2$-node, $v$ of label $p$ belongs to a successful path then it has a child of label $p+1$ after the successful merger. If $v$ were to participate in another successful merger at the same stage, the sequence of labels in this second successful path would be $\ell \rightarrow (\ell-1) \rightarrow .. \rightarrow (p+1) \rightarrow p \rightarrow (p+1) \rightarrow p \rightarrow .. \rightarrow 0$. The length of such a path is strictly greater than a two-edge distance $\ell$ since layer $p$ is visited twice in this path. Thus, each $V_2$-node participates in at most one successful path of length $\ell$ in stage $\ell$. Hence, a node that was part of a successful merger is marked as visited, and need not be visited again during that stage. \qed
\end{proof}

\begin{corollary}
\label{cor:nodethroughput}
Each $V_2$-node participates in at most one successful merger or one backtrack at each stage.
\end{corollary}

\begin{corollary}
\label{lem:increasinglabel}
The labels of all lowest labeled strong nodes of label $\ell$ at stage $\ell$ strictly increase after the termination of stage $\ell$.
\end{corollary}

Performing global relabeling is equivalent to generating a so-called {\em layered network} (as in Dinic's maximum flow algorithm). In a layered network, each layer consists of all branches with a particular label. In stage $\ell$, layer $0$ consists of all nodes that have distance label $0$, i.e., only $WT_1$ branches.  Layers $1$ through $\ell-1$ consist of $WT_2$ branches, and layer $\ell$ consists of $ST_1$ branches.

We now describe the procedure for generating the layered network, which is the critical part of our algorithm.  Let the {\em $k$-layer} $(0 \leq k \leq \ell)$ be the set of nodes with label $k$.  The layered network can be generated by scanning all backward residual arcs from the sink using a Breadth-First-Search (BFS).

In a naive implementation of BFS one would start with all $WT_1$ branches (label 0) and look at all incoming arcs in the residual network to generate the 1-layer. This could take $O(m)$ work and is expensive.   An alternative approach is to check for each $WT_2$ branch whether it is in the $1$-layer by checking if there is an arc from its child node to a $WT_1$ branch. Using the fact that labels are non-decreasing, we only need to check this for $WT_2$ branches that were of label $1$ in the previous layered network. \medskip

\noindent{\bf Generating the 1-layer:}~The neighbors of each $V_1$ child node of label $1$ in a $WT_2$ branch are scanned to identify a residual arc to a $WT_1$ branch. Since there are at most $\kappa$ $WT_2$ branches, we need to scan at most $\kappa$ neighbors of each $V_1$-node of label $1$ to identify a $WT_1$ neighbor node or determine that none exists (in which case the branch does not belong to the $1$-layer). If a $WT_1$ node is not adjacent to a $V_1$-node of label $1$ in stage $\ell$ then it cannot be adjacent to that $V_1$-node in any later stage.  This is because no new $WT_1$ branches are created in any stage.  Hence, each arc needs to be scanned at most once throughout the algorithm. By maintaining a pointer for each $V_1$-node to the last arc scanned at each stage, and resuming the search from that arc in the next stage, we can ensure that each arc is scanned at most once throughout the algorithm.

\begin{claim}
The total work done to generate the 1-layer {\em throughout the algorithm} is $O(\min\{\kappa^2, m\})$.
\end{claim}

\noindent{\bf Generating layers $2$ through $\ell-1$:}~Given the set of $WT_2$ branches in layer $p$, the incoming residual arcs into the root of each $WT_2$ branch in layer $p$ are examined to obtain neighbors in layer $p+1$. Scanning the incoming residual arcs of a $WT_2$ branch stops if an $ST_1$ neighbor is found, since then $p=\ell -1$ and the $WT_2$ branch is in the $\ell-1$ layer its $ST_1$ neighbor is thus in the $\ell$-layer.

There are at most $\kappa$ incoming arc scans for each $WT_2$ root required to label all the $WT_2$ branches in the next layer or find a $ST_1$ branch.   Since there are at most $\kappa$ roots of $WT_2$ branches and each arc is scanned at most once in each stage, so total work done in generating layers 2 through $\ell-1$ at each stage is $O(\min\{\kappa^2, m\})$.

\begin{claim}
The work done per stage to generate the layers $2$ through $\ell-1$ is $O(\min\{\kappa^2, m\})$.
\end{claim}

The layered network generated has all layers of weak branches up to the $\ell-1$-layer, and some $ST_1$ branches in the $\ell$-layer. This $\ell$-layer may not contain all the $ST_1$ branches of label $\ell$ since not all incoming arcs to the $WT_2$ branches were examined. However, by Corollary \ref{cor:nodethroughput}, it is sufficient to have at most one neighbor $ST_1$ branch for every $WT_2$ branch of label $\ell-1$. Therefore, instead of explicitly generating the entire $\ell$-layer, once an $ST_1$ branch of label $\ell$ is found, it is determined that the $WT_2$ branch is in the $\ell-1$-layer -- the last layer of weak branches.

For each unvisited $WT_2$ branch of label $\ell-1$, we scan its incoming arcs to check for an $ST_1$ neighbor. If such an $ST_1$ branch is found, a sequence of mergers is initiated from this strong branch. If a successful path is found, or if no more mergers are possible from this strong branch, another unvisited $WT_2$ branch in the $\ell-1$ layer is chosen and its incoming arcs are scanned to identify a new $ST_1$ branch from which mergers are initiated. This continues until all the $WT_2$ branches in the $\ell-1$ layer have been visited or have been scanned for a neighboring $ST_1$ branch.

There are at most $\kappa$ branches in the $\ell-1$ layer. For each such branch, at most $2 \kappa$ incoming arcs need to be scanned to identify a new neighboring $ST_1$ branch. Also, each arc in the network is examined at most once, so the work done per stage in identifying the necessary $ST_1$ branches in the $\ell$-layer is $O(\min\{m, \kappa^2\})$.

Each arc participates in at most one merger per stage, each of which requires $O(1)$ work; mergers thus require $O(\min\{\kappa^2, m\})$. A backtrack operation is performed at most once for each $WT_2$ branch, so work done in backtracking is $O(\kappa)$ per stage.

\begin{lemma}
The work done per stage including generating the layered network, mergers, and backtrack operations is $O(\min\{\kappa^2, m\})$.
\end{lemma}

\begin{lemma}
\label{lem:stages}
The number of stages in the algorithm is $O(\sqrt{\kappa})$.
\end{lemma}

The proof is along the lines of those of Even and Tarjan \cite{EveT75} and Hopcroft and Karp \cite{HopK73}.  Details are provided in Section \ref{proof:stages} of the appendix.

\begin{theorem}
For input given in the form of adjacency lists, the complexity of the {\sf matching-pseudoflow} algorithm is $O(\min\{n_1\kappa, m\} + \min\{\kappa^2,m\} \sqrt{\kappa})$.
\end{theorem}

\begin{proof}
The complexity of initialization is $O(\min\{n_1\kappa,m\})$. There are $O(\sqrt{\kappa})$ stages in the algorithm, each of which takes $O(\min\{\kappa^2, m\})$. The work to generate layer 1 is $O(\min\{\kappa^2,m\})$ throughout the algorithm. The total complexity is therefore $O(\min\{n_1\kappa, m\} + \min\{\kappa^2,m\} \sqrt{\kappa})$. \qed
\end{proof}

\noindent A high-level description of the {\sf matching-pseudoflow} algorithm is given in Figure \ref{algFigure:matchingPseudoflow}.

\algfigure{0.9}{htb}{ \Comment{The procedure finds a maximum matching in a bipartite graph $G = (V_1; V_2, E)$. It terminates with a set of $ST_1$, $WT_1$, and $WT_2$ branches; the set of edges in the $WT_2$ branches form a maximum cardinality matching.}\n
\Procedure matching-pseudoflow: \n
    \Begin \n
    Generate a greedy maximal matching of $ST_1$, $WT_1$, and $WT_2$ branches; \n
    Generate a layered network; \n
    Mark all nodes in the layered network as unvisited; \n
    \While the lowest label of an $ST_1$ branch is less than $|V_1|$ \Do \n
        \While $\exists$ a lowest labeled unvisited $V_1$-node $v$ of label $\ell$ \Do \n
            \If $\exists$ a merger from $v$ to an unvisited node of label $(\ell$-$1)$ \Do \n
                Perform merger (as in Figures \ref{fig:dfstree2}(a)--(b)); \n
                \If merger leads to an augmentation \Do \n
                    Mark all nodes along the successful path as visited; \EndLoop \EndLoop \n
            \Else \Do \n
                Mark branch containing node $v$ as visited; \n
                Perform backtrack (as in Figures \ref{fig:dfstree2}(c)--(d)); \EndLoop \EndLoop \n
        Generate a new layered network; \n
        Mark all nodes in the layered network as unvisited; \EndLoop \EndLoop \n
    \End \EndLoop
}{\label{algFigure:matchingPseudoflow}High-level description of the {\sf matching-pseudoflow} algorithm.}

Note that the {\sf matching-pseudoflow} algorithm could be viewed as an efficient implementation of Dinic's algorithm with two-edge pushes: a successful path of mergers is essentially an augmenting path, while the procedure for generating the layered network is the same once greedy initialization has been performed.  Similarly, the {\sf matching-pseudoflow} algorithm could also be interpreted as an implementation of push-relabel with two-edge pushes that uses delayed relabeling and global relabeling.

\subsection{Matching-pseudoflow with word operations}

The complexity of the matching-pseudoflow algorithm can be further
improved to $O\left (\min \{n_1\kappa, \frac{n_1n_2}{\lambda}, m\} +
\kappa^2 + \frac{\kappa^{2.5}}{\lambda}\right )$ using boolean word
operations, where $\lambda$ is the length of a word, as done by Cheriyan and Mehlhorn \cite{CheM96}.  The key
idea is to represent the graph adjacency structure using words and
performing boolean operations on these words to find merger arcs.
Details are provided in Section \ref{sec:words} of the appendix.

\subsection{A combined algorithm}
\label{sec:combined}

We follow the approach of Alt et al. \cite{AltBMP91} to combine the
{\sf matching-pseudoflow} with and without words to describe new
complexity bounds.  The new bound is obtained by applying the {\sf
matching-pseudoflow} without words until a certain  stage $\ell$ and
then using word operations for the rest of the algorithm. The greedy
initialization procedure is performed with words, which has a
complexity of \mbox{$O(\min\{n_1\kappa, \frac{n_1n_2}{\lambda},m\})$}.
The words SUB-IN and SUB-OUT described in section \ref{sec:words} are
also constructed irrespective of the algorithm used. These two
operations have a complexity of
\mbox{$O(\min\{n_1\kappa,\frac{n_1n_2}{\lambda},m\} + \kappa^2)$}.
\smallskip

\noindent {\bf Case 1: $\kappa^2 \in O(m)$} \smallskip

The work done until stage $\ell$ without using words is $O(\ell \kappa^2)$. Following analysis similar to that in the proof of Lemma \ref{lem:stages}, the remaining number of stages is $\kappa/\ell$.  The work done beyond stage $\ell$ using word operations is $O(\kappa^{2.5}/\lambda)$.  The value of $\ell$ that minimizes the total work done is obtained by solving for $\ell$ in the equation $\ell \kappa^2 = \frac{\kappa}{\ell} \frac{\kappa^{2}}{\lambda}$.

This yields $\ell = \sqrt{\kappa/\lambda}$ and an overall complexity of \mbox{$O(\min\{n_1\kappa,\frac{n_1n_2}{\lambda},m\} + \kappa^2 + \frac{\kappa^{2.5}}{\sqrt{\lambda}})$}, which is dominated by the complexity of the algorithm with word operations. Thus, when $\kappa^2 \in O(m)$, combining the two algorithms does not provide any benefit. \smallskip

\noindent {\bf Case 2: $\kappa^2 \in \Omega(m)$} \smallskip

The work done until stage $\ell$ is $O(\ell m)$.  We again find the best value of $\ell$ by solving $\ell m = \frac{\kappa}{\ell} \frac{\kappa^{2}}{\lambda}$, which gives $\ell = \kappa \sqrt{\frac{\kappa}{m \lambda}}$.  This leads to an overall complexity of \mbox{$O(\min\{n_1\kappa,\frac{n_1n_2}{\lambda},m\} + \kappa^2 + \kappa^{1.5}\sqrt{\frac{m}{\lambda}}))$}.

This is better than the $\sqrt{\kappa} m$ complexity when $\kappa^2 \in O(m \lambda)$.
Table \ref{Table:complexitySummary} summarizes the complexity results.
\begin{table}[ht]
\begin{center}
\begin{tabular}{|l|l|l|} \hline
{\bf Algorithm} & {\bf Best when} & {\bf Complexity} \\ \hline \hline
With word operations & $\kappa^2 \in O(m)$ & $O(\min\{n_1\kappa,\frac{n_1n_2}{\lambda},m\} + \kappa^2 + \kappa^{2.5}/\lambda)$\\ \hline
Combined & $\kappa^2 \in \Omega(m) \cup O(\lambda m)$ & $O(\min\{n_1\kappa,\frac{n_1n_2}{\lambda},m\} + \kappa^2 + \kappa^{1.5}\sqrt{\frac{m}{\lambda}})$ \\ \hline
Without word operations & $\kappa^2 \in \Omega(\lambda m)$ & $O\left (\min \{ n_1\kappa,m\} + \sqrt{\kappa}\min \{\kappa^2,m\}\right )$ \\ \hline
\end{tabular}
\caption{\label{Table:complexitySummary}Summary of complexity results for the {\sf matching-pseudoflow} algorithm.}
\end{center}
\end{table}

Note that the complexity expressions for the matching-pseudoflow algorithm without words, the combined algorithm, and the matching-pseudoflow algorithm with words are correspondingly faster than the algorithms of Hopcroft and Karp \cite{HopK73} with complexity $O(\sqrt{\kappa}m)$, Alt et al. \cite{AltBMP91} with complexity $O(n^{1.5} \sqrt{\frac{m}{\lambda}})$, and Cheriyan and Mehlhorn \cite{CheM96} with complexity $O(\frac{n^2.5}{\lambda})$.

\section{An experimental study}
\label{Sec:Expts}

\subsection{Implementations}

We developed eight pseudoflow implementations for bipartite matching:
\begin{enumerate}
\item Five ``regular'' pseudoflow implementations---highest label with FIFO buckets (pseudo\_hi\_fifo), highest label with LIFO buckets (pseudo\_hi\_lifo), highest label with Wave buckets (pseudo\_hi\_wave), lowest label with FIFO buckets (pseudo\_lo\_fifo), and lowest label with FIFO buckets (pseudo\_lo\_lifo).
\item Two ``free-arcs'' variants---{\sf pseudo\_hi\_free} and {\sf pseudo\_lo\_free} that are the highest and lowest label implementations of the free-arcs pseudoflow algorithm.  Both these implementation use LIFO buckets, which were found to be fastest in initial testing.  We use a global relabeling heuristic that periodically re-computes distance labels to all $V_1$-nodes in the graph.
\item The {\sf matching-pseudoflow} algorithm.
\end{enumerate}
The latest version of the code (version 1.01) is available at \cite{WebPS}.

Cherkassky et al. \cite{CheGMSS98} developed the following algorithms for bipartite matching that implement ``two-edge'' pushes:
\begin{itemize}
\item {\sf bim\_dfs} and {\sf bim\_bfs}: These two variants apply a simple depth-first-search and breadth-first-search respectively to find augmenting $s$-$t$ paths.
\item {\sf pr\_bim\_hi}, {\sf pr\_bim\_lo}, and {\sf pr\_bim\_fifo}: These are implementations of the highest label, lowest label, and FIFO push-relabel variants respectively.
\item {\sf bim\_ar}: The ``augment-relabel'' algorithm could be thought of as a hybrid between an augmenting path algorithm and push-relabel. It is similar in spirit to the basic algorithm described by Alt et al. \cite{AltBMP91} for the bipartite matching problem.
\item {\sf bim\_lds}: The ``label-directed-search'' variant uses a depth-first-search along with ``approximate'' distance labels that are periodically updated using global relabeling.
\end{itemize}

In addition, we tested {\sf dinic}, an implementation of Dinic's algorithm by Setubal \cite{Set93}, and {\sf abmp}, a simplified implementation by Setubal \cite{Set96} of the algorithm of Alt et al.\cite{AltBMP91} that is available as part of the BIPM solvers for bipartite matching \cite{BIPM}.  While {\sf dinic} was shown to have poor performance in practice, we use it mainly to compare it to {\sf matching-pseudoflow}, which is its closest pseudoflow counterpart.

The pseudoflow codes, {\sf dinic}, and {\sf abmp} were written in C and compiled with the {\tt gcc} compiler while those of Cherkassky et al. \cite{CheGMSS98} were written in C++ and compiled used the {\tt g++} compiler.  The {\tt -O4} compiler optimization flag was used in all cases.

\subsection{Computing environment}

The experiments were run on a Sun UltraSPARC workstation with a 270 MHz CPU and 192 MB of RAM.  The results of the machine calibration experiment as suggested by the First Dimacs Implementation Challenge \cite{Dimacs90} are shown in Table \ref{Table:machineCalib}.

\begin{table}[ht]
\begin{center}
\begin{tabular}{||c|ccc|ccc||}
\hline \hline
{} & \multicolumn{3}{c|}{Test 1} & \multicolumn{3}{c||}{Test 2}\\
{} & {real} & {user} & {system} & {real} & {user} & {system}\\
\hline
No optimization & 0.4 & 0.4 & 0.0 & 3.3 & 3.3 & 0.0\\
-O4 flag & 0.2 & 0.1 & 0.0 & 2.0 & 1.9 & 0.0\\
\hline \hline
\end{tabular}
\caption{\label{Table:machineCalib}Average running times for Dimacs machine calibration tests.}
\end{center}
\end{table}

\subsection{Differences between {\sf matching-pseudoflow} and {\sf dinic} in practice}

As noted earlier, the {\sf matching-pseudoflow} algorithm has parallels
to Dinic's algorithm (global relabeling is equivalent to generating a
layered network in Dinic's algorithm, and a successful path is
equivalent to an $s$-$t$ augmenting path).  However, what sets the {\sf
matching-pseudoflow} algorithm apart from Dinic's algorithm is the
manner in which global relabeling is performed.  In this section, we
demonstrate that the global relabeling procedure which leads to a
better theoretical complexity also makes a significant difference in
practice.

In the {\sf matching-pseudoflow}, only the nodes in the current
matching and their adjacent edges are examined during global
relabeling, whereas in Dinic's algorithm the entire network (including
all unmatched $V_2$-nodes) and their adjacent arcs are examined to
construct the layered network. Therefore, in practice, we would expect
the run-time of the {\sf matching-pseudoflow} algorithm to be dependent
largely on $\kappa$, while that of Dinic's algorithm to be dependent on
$m$.

We implemented the {\sf matching-pseudoflow} algorithm and compared it to the best-known implementation of Dinic's algorithm for bipartite matching \cite{Set93}.  Instances were generated in the following manner: given $n_1$, $n_2$, and the expected number of edges $\overline{m}$,  a graph is generated where each of the possible $n_1 n_2$ edges exists independently with probability $\overline{m}/(n_1 n_2)$.  We generated problems with $n_1=$ 16384, and $n_2/n_1$ ranging from 1 to 1.2 in steps of 0.01.  Thus, the most unbalanced graph had 19661 nodes.  For each of the 20 classes, we generated graphs with expected number of edges 81920, 163840, 245760, and 327680 respectively, which resulted in $\kappa$ being exactly or very close to $n_1$.  For each of the combinations of $n_2$ and $m$, we generated 10 instances and the time for each instance was averaged over 5 runs. Thus, each data point is the average of 50 runs.  The run-times are shown in Figure \ref{fig:unbalanced}.

\begin{figure}[ht]
\begin{minipage}{0.49\linewidth}
\begin{center}
\centerline{\includegraphics[width=\linewidth]{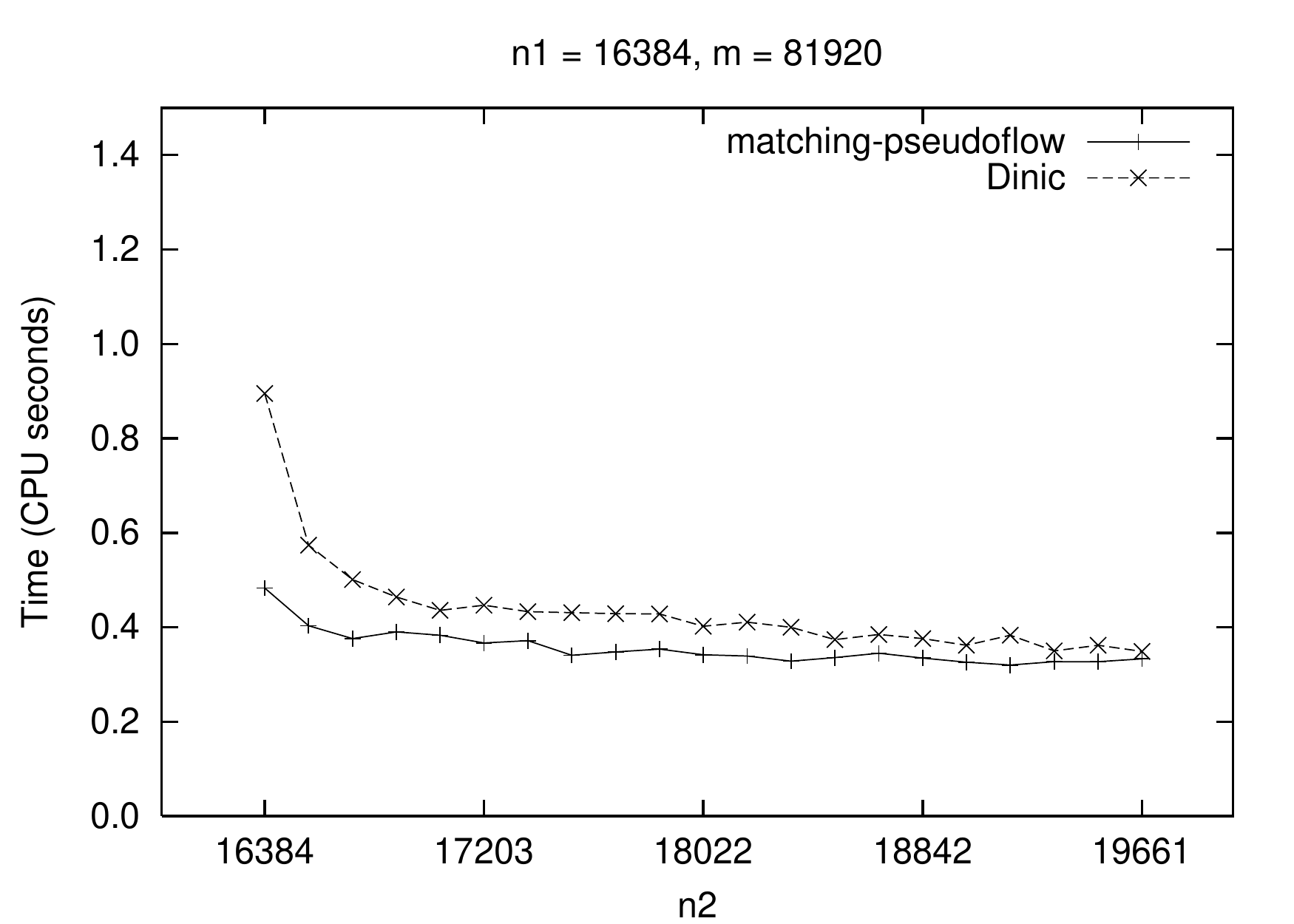}}
\centerline{\includegraphics[width=\linewidth]{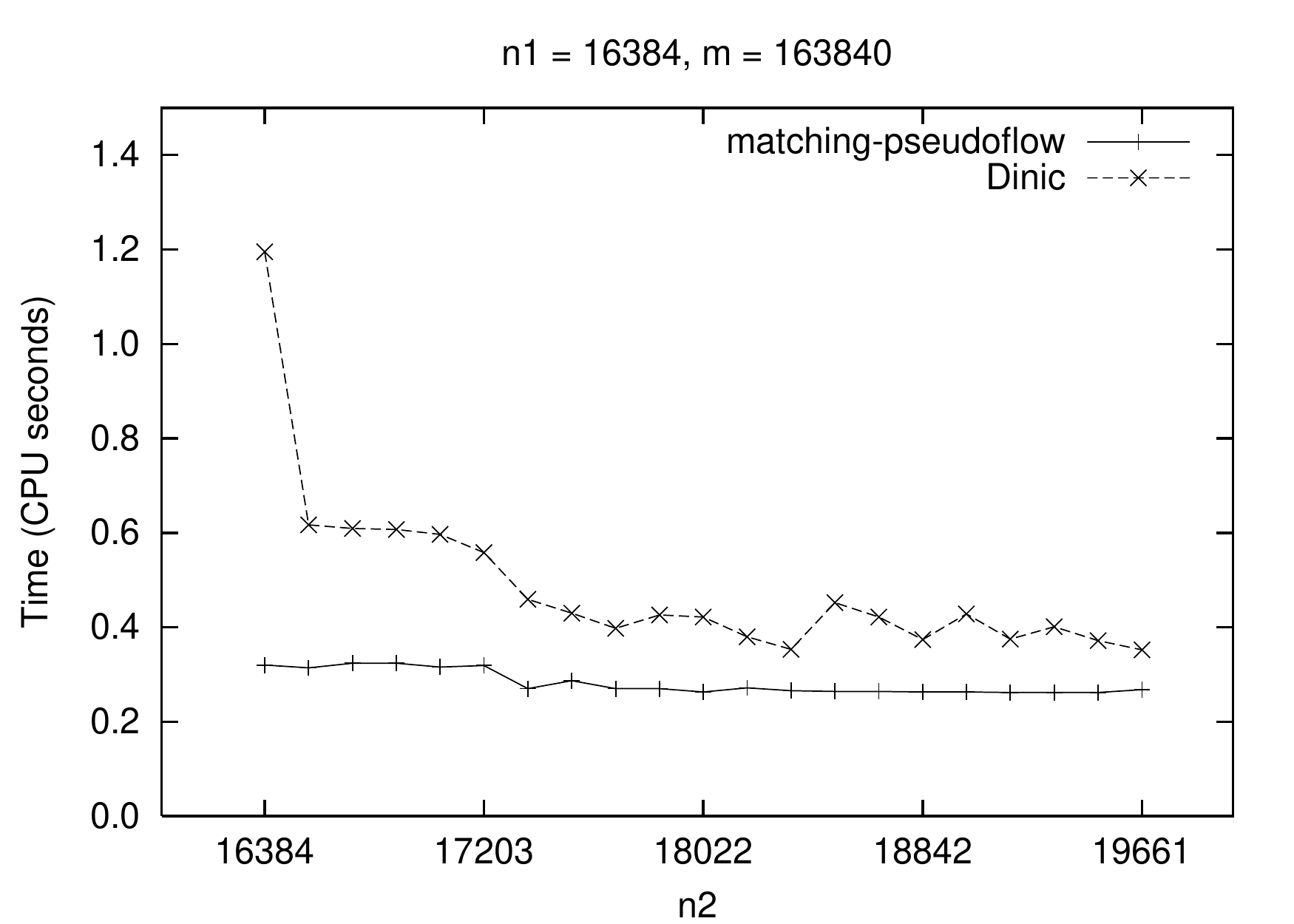}}
\end{center}
\end{minipage}\hfill
\begin{minipage}{0.49\linewidth}
\begin{center}
\centerline{\includegraphics[width=\linewidth]{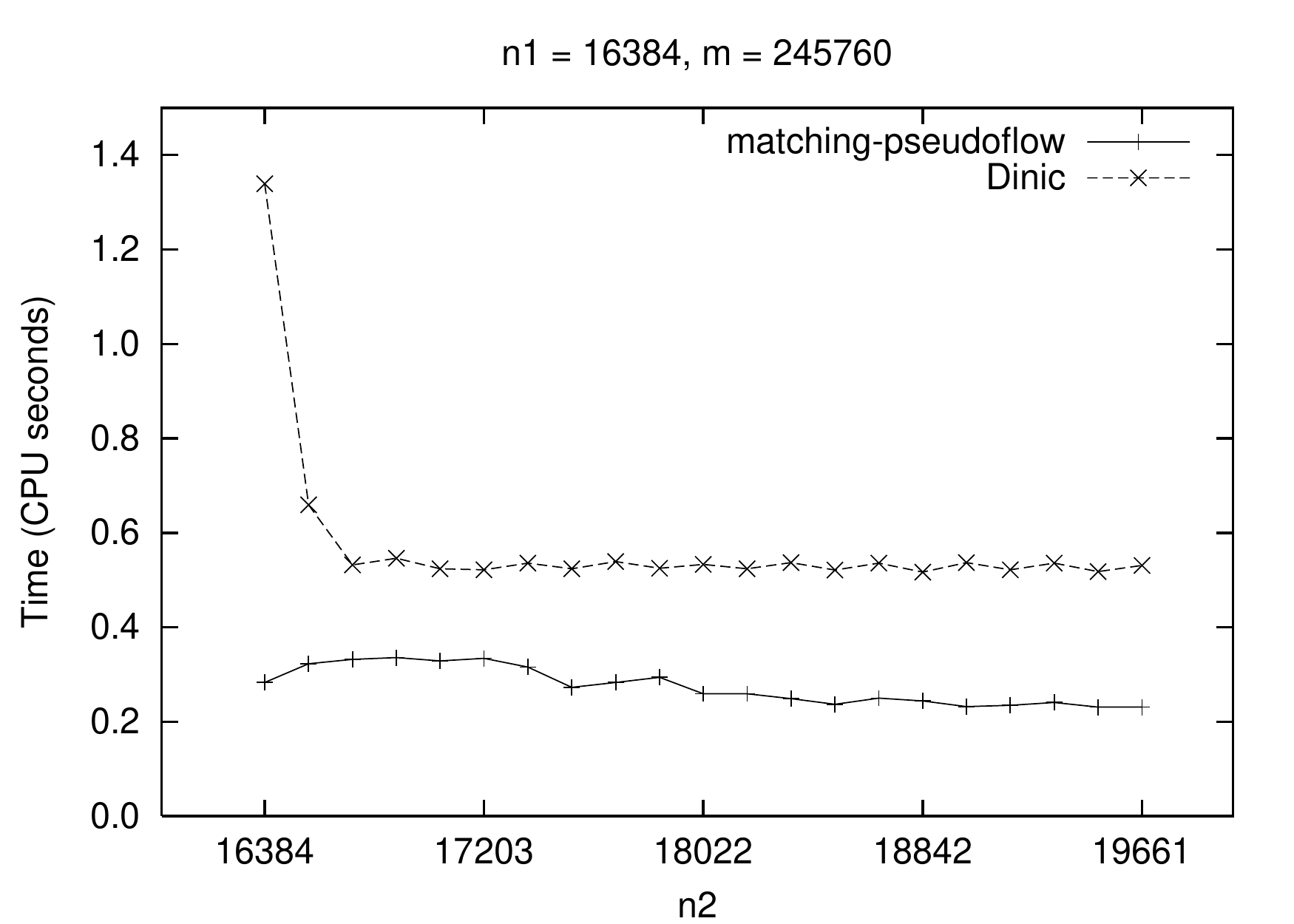}}
\centerline{\includegraphics[width=\linewidth]{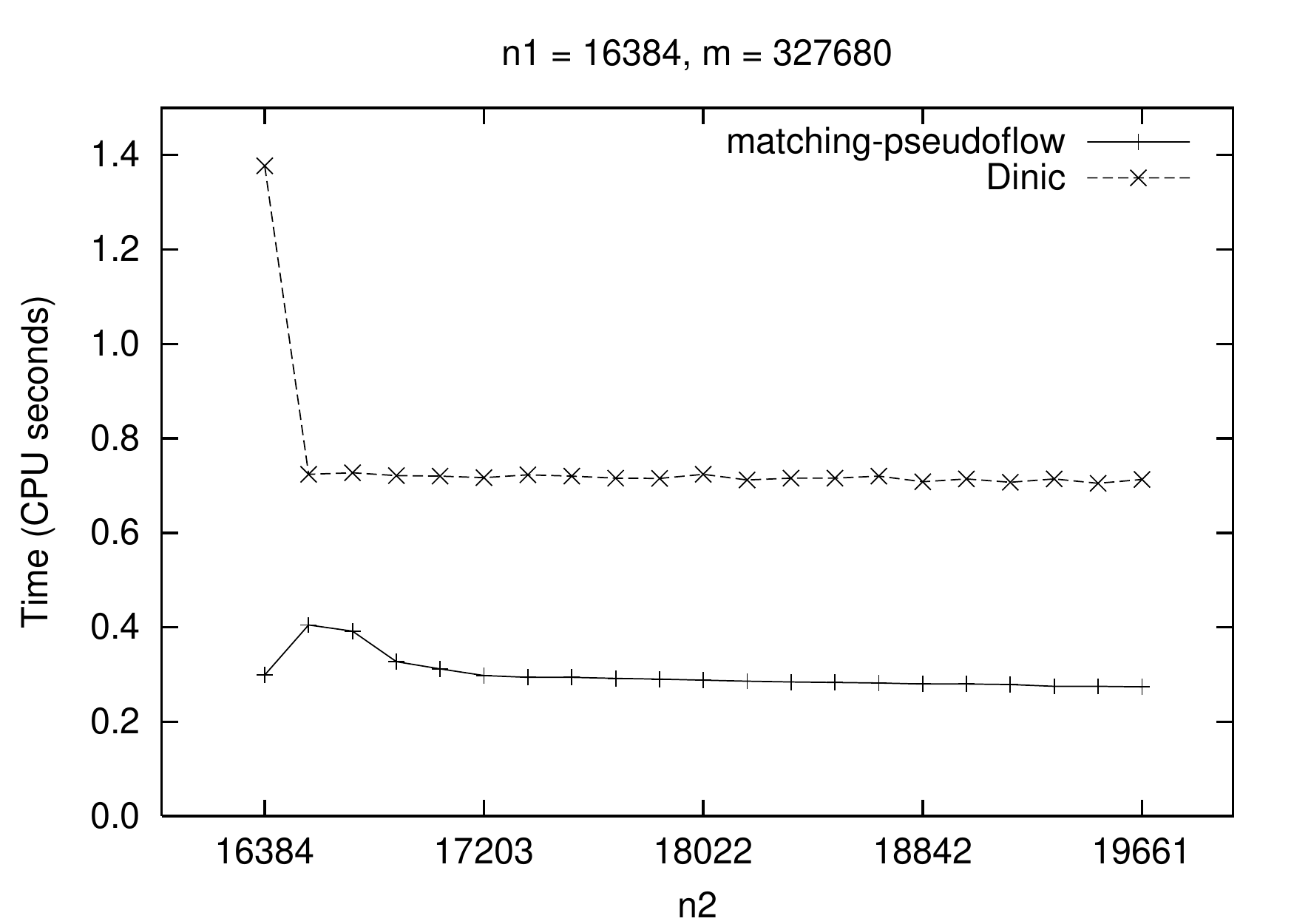}}
\end{center}
\end{minipage}
\caption{\label{fig:unbalanced} Run-times of the {\sf matching-pseudoflow} and Dinic's algorithms on random unbalanced instances.}
\end{figure}

There are two key observations to be made from the results.
\begin{enumerate}
\item The {\sf matching-pseudoflow} algorithm is more robust to imbalances in the graph.  For Dinic's algorithm, the balanced instances are the hardest to solve; even small imbalances in the graph ($n_2 = 1.01 n_1$) make drastic differences to the run-time.
\item The run-time of Dinic's algorithm goes up with the number of edges, but the run-time of the {\sf matching-pseudoflow} is virtually independent of the number of edges. In fact, the hardest instances for the {\sf matching-pseudoflow} appear to be the ones with fewest arcs.
\end{enumerate}

\subsection{Test instances}

We tested the algorithms on the seven problem families ({\sf hilo},
{\sf fewg}, {\sf manyg}, {\sf grid}, {\sf hexa}, {\sf rope}, and {\sf
zipf}) used by Cherkassky et al. \cite{CheGMSS98}. All the benchmark
instances were balanced, i.e., $n_1 = n_2$.  The instances are
described in greater detail in Section \ref{sec:instances} of the
appendix.

For each instance family, we report the results of our experiments for
\begin{itemize}
\item Five pseudoflow implementations: {\sf pseudo\_lo\_lifo} and {\sf pseudo\_hi\_wave}, which were found in initial testing to be the fastest variants for the lowest and highest label algorithms respectively, {\sf pseudo\_lo\_free}, {\sf pseudo\_hi\_free}, and {\sf matching-pseudoflow}.  All the pseudoflow variants were initialized with a greedy matching.
\item Three implementations of Cherkassky et al. \cite{CheGMSS98}: {\sf pr\_bim\_hi}, {\sf pr\_bim\_lo}, and the best implementation among {\sf pr\_bim\_fifo}, {\sf bim\_dfs}, {\sf bim\_bfs}, {\sf bim\_ar}, and {\sf bim\_lds}. The {\sf pr\_bim\_hi} and {\sf pr\_bim\_lo} implementation were tested on all families to compare them to the free-arcs pseudoflow variants.
\item Implementations {\sf abmp} and {\sf dinic}.
\end{itemize}

\subsection{Results}

\begin{itemize}
\item {\bf Hi-lo:~}  The run-times and operation counts for {\sf hilo} instances are presented in Figure \ref{Figure:hiloscaling} and Table \ref{Table:hiloopcount} respectively.

The {\sf hilo} family was designed to be much harder for the highest label push-relabel algorithm than the lowest label variant.  As expected, {\sf pseudo\_hi\_free} and {\sf bim\_hi\_free} are the slowest, though the former is faster than the latter.  The {\sf pseudo\_lo\_free} variant is the fastest of all algorithms, and is more than twice as fast as {\sf bim\_lo\_free}.

Interestingly, the {\sf pseudo\_hi\_wave} is faster than {\sf pseudo\_lo\_lifo}, showing once again that pseudoflow and push-relabel have very different behavior.

The {\sf pseudo\_hi\_wave} and  {\sf bim\_bfs} algorithms show the best scaling behavior and are likely to be faster than {\sf pseudo\_lo\_free} on larger instances.  The {\sf matching-pseudoflow} algorithm shows poor scaling behavior; it is faster than {\sf dinic} on smaller instances but becomes slower on large instances.

The {\sf pseudo\_hi\_wave} algorithm performs fewer arc scans and pushes (the dominant operations) than {\sf pseudo\_lo\_free}, yet is slower.  This suggests that the simplicity of the free-arcs implementations result in performance gains due to simplicity of code (which often leads to better compiler optimization).

\item {\bf Fewg:~}  The run-times and operation counts for {\sf fewg} instances are presented in Figure \ref{Figure:fewgscaling} and Table \ref{Table:fewgopcount} respectively.

The {\sf pseudo\_hi\_free} and {\sf pseudo\_lo\_free} algorithms are the fastest, and are more than twice as fast as the next-best algorithms ({\sf pr\_bim\_hi} and {\sf pr\_bim\_lo}).  The difference seems to be in the number of arc scans performed.

The {\sf matching-pseudoflow} and {\sf dinic} algorithms are the slowest, though {\sf matching-pseudoflow} is faster on all instance sizes.

\item {\bf Manyg:~}  The run-times and operation counts for {\sf manyg} instances are presented in Figure \ref{Figure:manygscaling} and Table \ref{Table:manygopcount} respectively.

The results are similar to the {\sf fewg} instances.  The {\sf pseudo\_hi\_free} and {\sf pseudo\_lo\_free} algorithms are the fastest, and are more than twice as fast as the next-best implementations ({\sf pr\_bim\_hi} and {\sf pr\_bim\_lo}), which is reflected in the number of arc scans performed.

While the {\sf matching-pseudoflow} implementation is faster than {\sf dinic} on all instance sizes, {\sf dinic} appears to scale better and is likely to be faster on larger instances.

\item {\bf Grid:~}  The run-times and operation counts for {\sf grid} instances are presented in Figure \ref{Figure:gridscaling} and Table \ref{Table:gridopcount} respectively.

The scaling behavior of all the pseudoflow variants is extremely non-robust, making a comparison of the algorithms difficult.  However, the {\sf pseudo\_hi\_free} variant is the fastest on all instance sizes with the {\sf pseudo\_lo\_free} variant close behind.  These variants are more than twice as fast as the next-best algorithms ({\sf pr\_bim\_hi} and {\sf pr\_bim\_lo}).

The {\sf matching-pseudoflow} algorithm is faster than {\sf dinic}, although its scaling behavior is not robust.

The {\sf pseudo\_hi\_free}, {\sf pseudo\_lo\_free}, {\sf pr\_bim\_hi} and {\sf pr\_bim\_lo} algorithms did not perform any global relabeling.  Hence, this would be a good family to understand the fundamental differences between the four implementations.  We see that the push-relabel variants perform a greater number of each of the operations; however, it is difficult to draw strong conclusions due to the non-robust scaling behavior of the pseudoflow variants.

\item {\bf Hexa:~}  The run-times and operation counts for {\sf hexa} instances are presented in Figure \ref{Figure:hexascaling} and Table \ref{Table:hexaopcount} respectively.

The {\sf pseudo\_hi\_free} and {\sf pseudo\_lo\_free} algorithms are the fastest, followed by {\sf pr\_bim\_lo} which is 1.5--1.8 times slower, which is reflected in the number of arc scans performed.

The {\sf matching-pseudoflow} algorithm is faster than {\sf dinic} by a similar factor.

\item {\bf Rope:~}  The run-times and operation counts for {\sf rope} instances are presented in Figure \ref{Figure:ropescaling} and Table \ref{Table:ropeopcount} respectively.

This was the only family where {\sf abmp} showed good performance, and is the fastest of all algorithms.  The {\sf matching-pseudoflow} is only marginally slower.  Both {\sf matching-pseudoflow} and {\sf pseudo\_lo\_free} scale better than {\sf abmp} and are likely to be faster on larger instances.   The {\sf matching-pseudoflow} algorithm is much faster than {\sf dinic}, while the {\sf pseudo\_hi\_free} algorithm is an order of magnitude faster than {\sf pr\_bim\_hi}.

The operation counts do not provide much insight.

\item {\bf Zipf:~}  The run-times and operation counts for {\sf zipf} instances are presented in Figure \ref{Figure:zipfscaling} and Table \ref{Table:zipfopcount} respectively.

The {\sf pseudo\_hi\_free} algorithm is the fastest, with {\sf matching-pseudoflow} close behind.  The next best algorithm is {\sf pseudo\_lo\_free} (note that this is the only family in which {\sl pseudo\_lo\_free}) is not the best or nearly best algorithm.

The difference between the highest and lowest label variants seems to be due to the fact that no global relabels are triggered in the highest label variant, while the lowest label variants perform one relabel.
\end{itemize}
\section{Discussion}

We developed several variants of the pseudoflow algorithm for bipartite matching. One variant, the {\sf matching-pseudoflow} algorithm was shown to have the best-known theoretical complexity for the problem.  While the {\sf  matching-pseudoflow} could be viewed as a specialized implementation of Dinic's algorithm, we believe that the {\sf matching-pseudoflow} is a natural extension of the generic pseudoflow algorithm, whereas Dinic's algorithm requires a greater degree of adaptation from its widely-accepted form.  We also compared the {\sf matching-pseudoflow} to Dinic's algorithm to point out the key differences between the two algorithms. 

We also developed several implementations of our algorithms and compared them to the fastest available codes based on the push-relabel algorithm. We draw the following conclusions from our experiments.
\begin{itemize}
\item Our best implementation was faster than that of Cherkassky et al. \cite{CheGMSS98} on each problem family tested.  The {\sf psuedo\_lo\_free} algorithm was the fastest or nearly fastest algorithm is six of the seven instance classes tested.  On the remaining family ({\sf zipf}), it was the third-fastest implementation and was within a factor of 2 of the fastest implementation.  We hence declare this to be the best pseudoflow variant overall and recommend that it be the algorithm of choice when solving bipartite matching problems.
\item The {\sf pseudo\_lo\_free} variant was generally faster than the {\sf pseudo\_hi\_free} variant.  This is consistent with the behavior of push-relabel where the lowest label variant was found to be faster than the highest label variant.  However, in the regular pseudoflow variant (without free arcs), the highest label variant was generally faster than the lowest label variant.
\item While the {\sf psuedo\_lo\_free} and {\sf pseudo\_hi\_free} could be viewed as special implementations of the push-relabel algorithm with a two-edge push, they are uniformly faster than the push-relabel implementations of Cherkassky et al. \cite{CheGMSS98}.

This difference is not due only to the different global relabeling frequency.  In the {\sf grid} instances where no global relabeling was performed, push-relabel variants performed more operations such as arc scans and pushes than the pseudoflow variants.
\item Although implementations based on the regular pseudoflow algorithm (i.e., without free arcs) were faster than push-relabel for unit capacity networks \cite{Cha07}, their performance is unimpressive for bipartite matching.  This is surprising given that bipartite matching is a special case of unit capacity networks.  In general, {\sf pseudo\_hi\_wave} and {\sf pseudo\_lo\_lifo} were at least a factor of 2 slower than the fastest algorithm.  However, their performance was comparable to that of {\sf pr\_bim\_hi} and {\sf pr\_bim\_lo} on four of the families.
\item The {\sf matching-pseudoflow} algorithm is generally faster than {\sf dinic}, and is nearly the fastest algorithm on two instance families.  This is particularly interesting because the {\sf matching-pseudoflow} algorithm could be viewed as an efficient implementation of Dinic's algorithm.  Past experimental studies \cite{Set93, Set96, CheGMSS98} have dismissed Dinic's algorithm as not being competitive in practice.  However, the results here show that a careful implementation of Dinic's algorithm (i.e., the {\sf matching-pseudoflow}) can be very efficient in practice.
\item The experiments comparing {\sf matching-pseudoflow} and {\sf dinic} on random graphs clearly shows that the theoretically efficient global relabeling procedure is efficient in practice as well.

On benchmark instances, {\sf matching-pseudofow} often performed a much greater number of global relabels.  This is because {\sf matching-pseudoflow} generates the layered network only until the lowest labeled layer of excess nodes and finds a blocking flow in this network, while {\sf dinic} creates a layered network consisting of {\em all} excess nodes in the network and finds a blocking flow in this network.
\end{itemize}

\newpage
\clearpage

\begin{figure}[ht]
\begin{center}
\includegraphics[width = \linewidth]{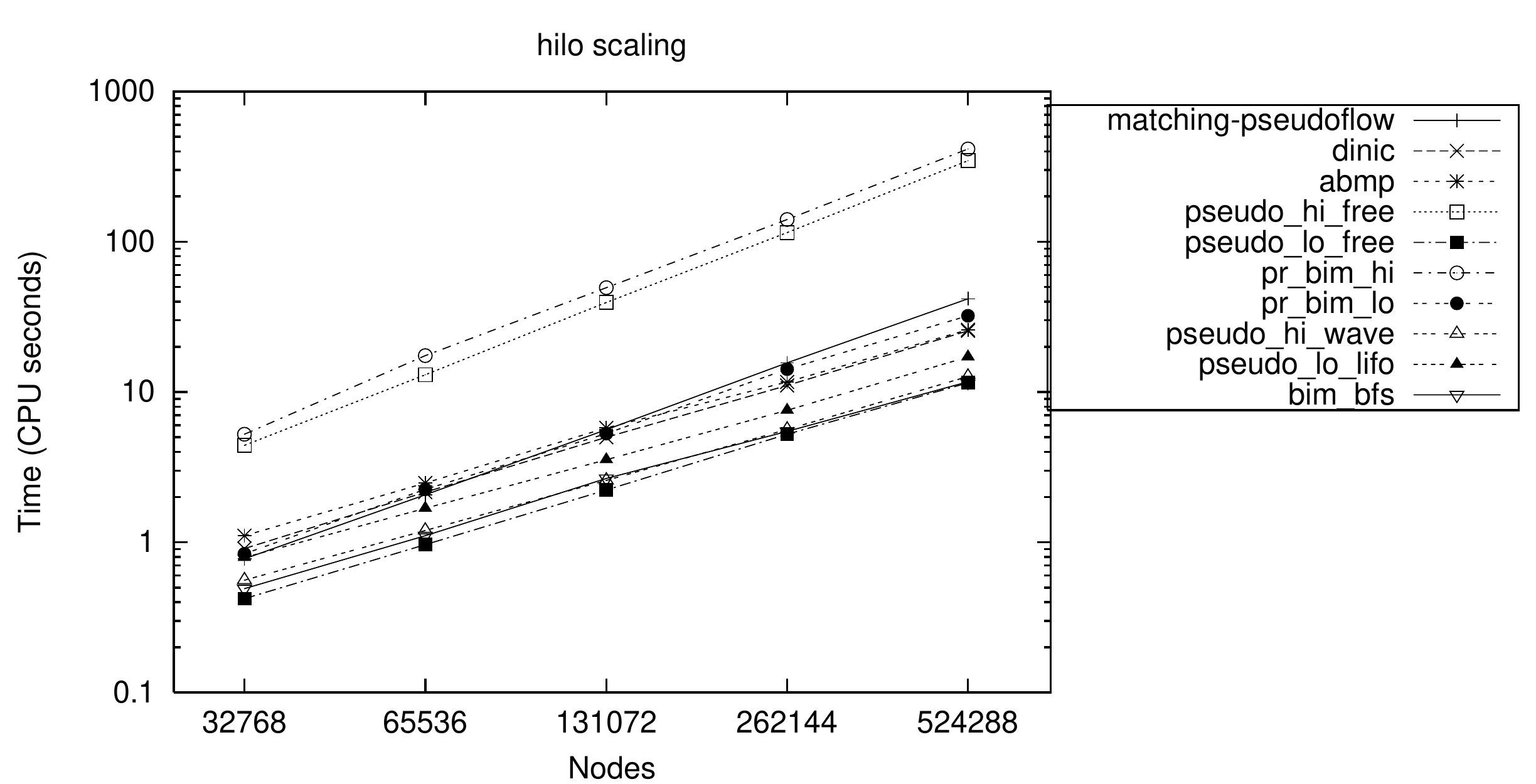}\bigskip
\begin{scriptsize}
\begin{tabular}{||c|r|r|r|r|r||}
\hline \hline
$n$, $m$ & \multicolumn{5}{|c||}{} \\ \hline
\multirow{6}{0.45in}{32,768 \\ 163,795} &   \multicolumn{1}{|c|}{{\sf pseudo\_hi\_wave}}    &   \multicolumn{1}{|c|}{{\sf pseudo\_lo\_lifo}}    &   \multicolumn{1}{|c|}{{\sf pseudo\_hi\_free}}    &   \multicolumn{1}{|c|}{{\sf pseudo\_lo\_free}}    &   \multicolumn{1}{|c||}{{\sf matching-pseudoflow}}    \\  \cline{2-6}
    &   0.559   &   0.791   &   4.413   &   0.422   &   0.776   \\
    &   (1.325) &   (1.877) &   (10.468)    &   ({\bf 1.000})   &   (1.841) \\  \cline{2-6}
    &   \multicolumn{1}{|c|}{{\sf abmp}}    &   \multicolumn{1}{|c|}{{\sf dinic}}   &   \multicolumn{1}{|c|}{{\sf pr\_bim\_hi}} &   \multicolumn{1}{|c|}{{\sf pr\_bim\_lo}} &   \multicolumn{1}{|c||}{{\sf bim\_bfs}}   \\  \cline{2-6}
    &   1.106   &   0.906   &   5.238   &   0.838   &   0.492   \\
    &   (2.623) &   (2.150) &   (12.424)    &   (1.987) &   (1.166) \\  \hline
\multirow{6}{0.45in}{65,536 \\ 327,635} &   \multicolumn{1}{|c|}{{\sf pseudo\_hi\_wave}}    &   \multicolumn{1}{|c|}{{\sf pseudo\_lo\_lifo}}    &   \multicolumn{1}{|c|}{{\sf pseudo\_hi\_free}}    &   \multicolumn{1}{|c|}{{\sf pseudo\_lo\_free}}    &   \multicolumn{1}{|c||}{{\sf matching-pseudoflow}}    \\
\hline
\hline
\end{tabular}
\end{scriptsize}
\caption{\label{Figure:hiloscaling} Actual and relative run times for {\sf hilo} instances.}
\end{center}
\end{figure}

\begin{table}[ht]
\begin{center}
\begin{scriptsize}
\begin{tabular}{||c|l|r|r|r|r|r||}
\hline \hline
$n$, $m$ & \multicolumn{6}{|c||}{} \\ \hline
\multirow{10}{0.45in}{32,768 \\ 163,795}    &       &   \multicolumn{1}{|c|}{{\sf pseudo\_hi\_wave}}    &   \multicolumn{1}{|c|}{{\sf pseudo\_lo\_lifo}}    &   \multicolumn{1}{|c|}{{\sf pseudo\_hi\_free}}    &   \multicolumn{1}{|c|}{{\sf pseudo\_lo\_free}}    &   \multicolumn{1}{|c||}{{\sf matching-pseudoflow}}    \\  \cline{2-7}
    &   Arc scans   &   1,056,128   &   1,646,324   &   8,192,128   &   762,861 &   \multicolumn{1}{|c||}{--}   \\
    &   Pushes  &   148,981 &   217,719 &   1,667,353   &   174,689 &   \multicolumn{1}{|c||}{--}   \\
    &   Relabels    &   237,005 &   380,975 &   804,735 &   61,909  &   \multicolumn{1}{|c||}{--}   \\
    &   Updates &   \multicolumn{1}{|c|}{--}    &   \multicolumn{1}{|c|}{--}    &   116 &   4   &   36  \\
    &   Mergers &   71,035  &   70,219  &   841,869 &   95,537  &   \multicolumn{1}{|c||}{--}   \\
    &   Depth   &   2.1 &   3.1 &   \multicolumn{1}{|c|}{--}    &   \multicolumn{1}{|c|}{--}    &   \multicolumn{1}{|c||}{--}   \\
\hline
\hline
\end{tabular}
\end{scriptsize}
\caption{\label{Table:hiloopcount}Operation counts for {\sf hilo} instances.}
\end{center}
\end{table}

\clearpage
\newpage

\begin{figure}[ht]
\begin{center}
\includegraphics[width = \linewidth]{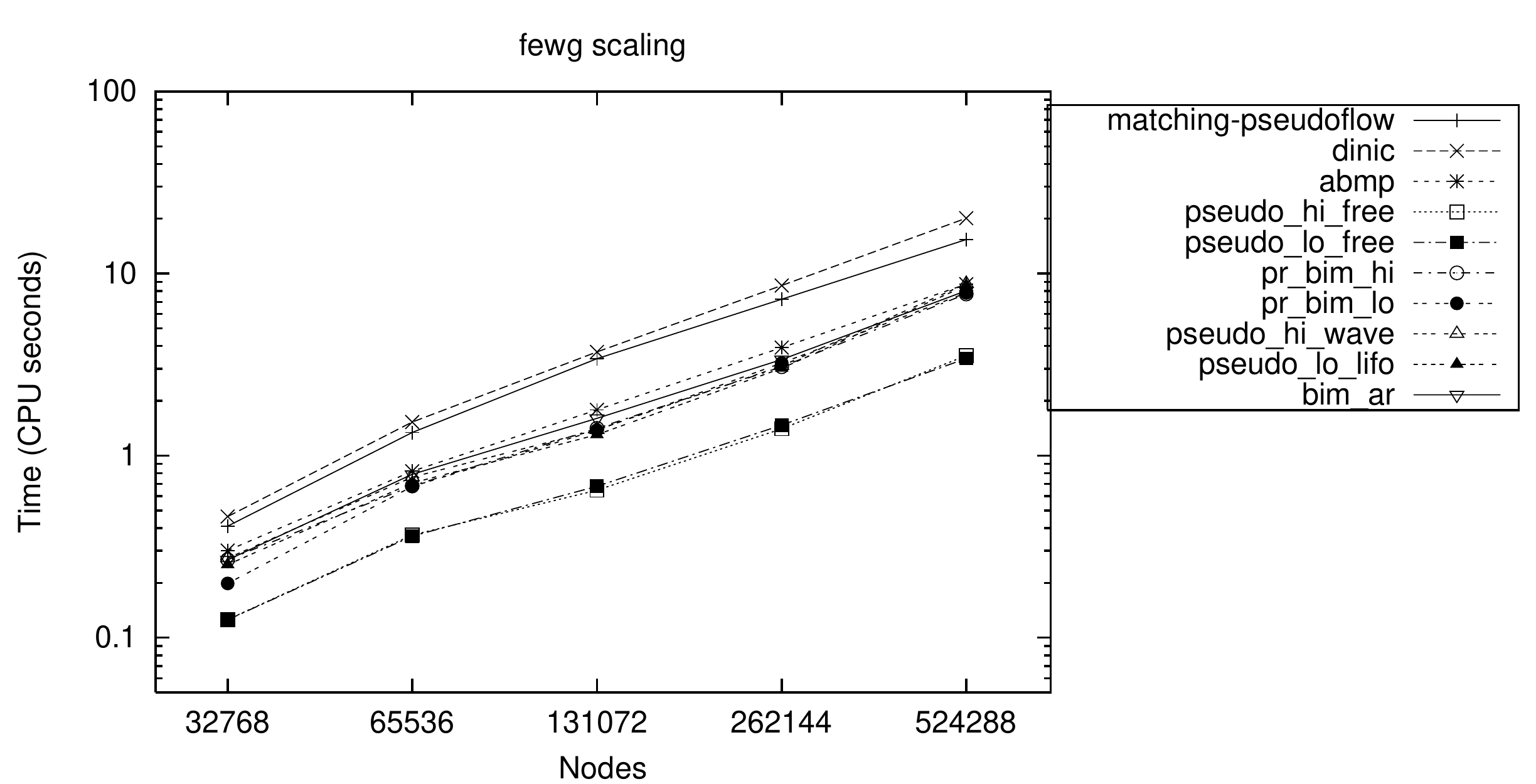}\bigskip
\begin{scriptsize}
\begin{tabular}{||c|r|r|r|r|r||}
\hline \hline
$n$, $m$ & \multicolumn{5}{|c||}{} \\ \hline
\multirow{6}{0.45in}{32,768 \\ 82,095}  &   \multicolumn{1}{|c|}{{\sf pseudo\_hi\_wave}}    &   \multicolumn{1}{|c|}{{\sf pseudo\_lo\_lifo}}    &   \multicolumn{1}{|c|}{{\sf pseudo\_hi\_free}}    &   \multicolumn{1}{|c|}{{\sf pseudo\_lo\_free}}    &   \multicolumn{1}{|c||}{{\sf matching-pseudoflow}}    \\  \cline{2-6}
    &   0.2736  &   0.2516  &   0.1258  &   0.1256  &   0.4104  \\
    &   (2.178) &   (2.003) &   (1.002) &   ({\bf 1.000})   &   (3.268) \\  \cline{2-6}
    &   \multicolumn{1}{|c|}{{\sf abmp}}    &   \multicolumn{1}{|c|}{{\sf dinic}}   &   \multicolumn{1}{|c|}{{\sf pr\_bim\_hi}} &   \multicolumn{1}{|c|}{{\sf pr\_bim\_lo}} &   \multicolumn{1}{|c||}{{\sf bim\_ar}}    \\  \cline{2-6}
    &   0.3012  &   0.4628  &   0.268   &   0.1988  &   0.2674  \\
    &   (2.398) &   (3.685) &   (2.134) &   (1.583) &   (2.129) \\  \hline
\multirow{6}{0.45in}{65,536 \\ 163,656} &   \multicolumn{1}{|c|}{{\sf pseudo\_hi\_wave}}    &   \multicolumn{1}{|c|}{{\sf pseudo\_lo\_lifo}}    &   \multicolumn{1}{|c|}{{\sf pseudo\_hi\_free}}    &   \multicolumn{1}{|c|}{{\sf pseudo\_lo\_free}}    &   \multicolumn{1}{|c||}{{\sf matching-pseudoflow}}    \\
\hline
\hline
\end{tabular}
\end{scriptsize}
\caption{\label{Figure:fewgscaling} Actual and relative run times for {\sf fewg} instances.}
\end{center}
\end{figure}

\begin{table}[ht]
\begin{center}
\begin{scriptsize}
\begin{tabular}{||c|l|r|r|r|r|r||}
\hline \hline
$n$, $m$ & \multicolumn{6}{|c||}{} \\ \hline
\multirow{10}{0.45in}{32,768 \\ 82,095} &       &   \multicolumn{1}{|c|}{{\sf pseudo\_hi\_wave}}    &   \multicolumn{1}{|c|}{{\sf pseudo\_lo\_lifo}}    &   \multicolumn{1}{|c|}{{\sf pseudo\_hi\_free}}    &   \multicolumn{1}{|c|}{{\sf pseudo\_lo\_free}}    &   \multicolumn{1}{|c||}{{\sf matching-pseudoflow}}    \\  \cline{2-7}
    &   Arc scans   &   181,799 &   172,026 &   204,682 &   196,601 &   \multicolumn{1}{|c||}{--}   \\
    &   Pushes  &   161,912 &   170,298 &   88,909  &   85,202  &   \multicolumn{1}{|c||}{--}   \\
    &   Relabels    &   85,772  &   80,903  &   35,429  &   34,384  &   \multicolumn{1}{|c||}{--}   \\
    &   Updates &   \multicolumn{1}{|c|}{--}    &   \multicolumn{1}{|c|}{--}    &   1   &   1   &   13  \\
    &   Mergers &   34,719  &   33,472  &   52,593  &   50,739  &   \multicolumn{1}{|c||}{--}   \\
    &   Depth   &   4.7 &   5.1 &   1.7 &   1.7 &   \multicolumn{1}{|c||}{--}   \\
\hline
\hline
\end{tabular}
\end{scriptsize}
\caption{\label{Table:fewgopcount}Operation counts for {\sf fewg} instances.}
\end{center}
\end{table}

\clearpage
\newpage

\begin{figure}[ht]
\begin{center}
\includegraphics[width = \linewidth]{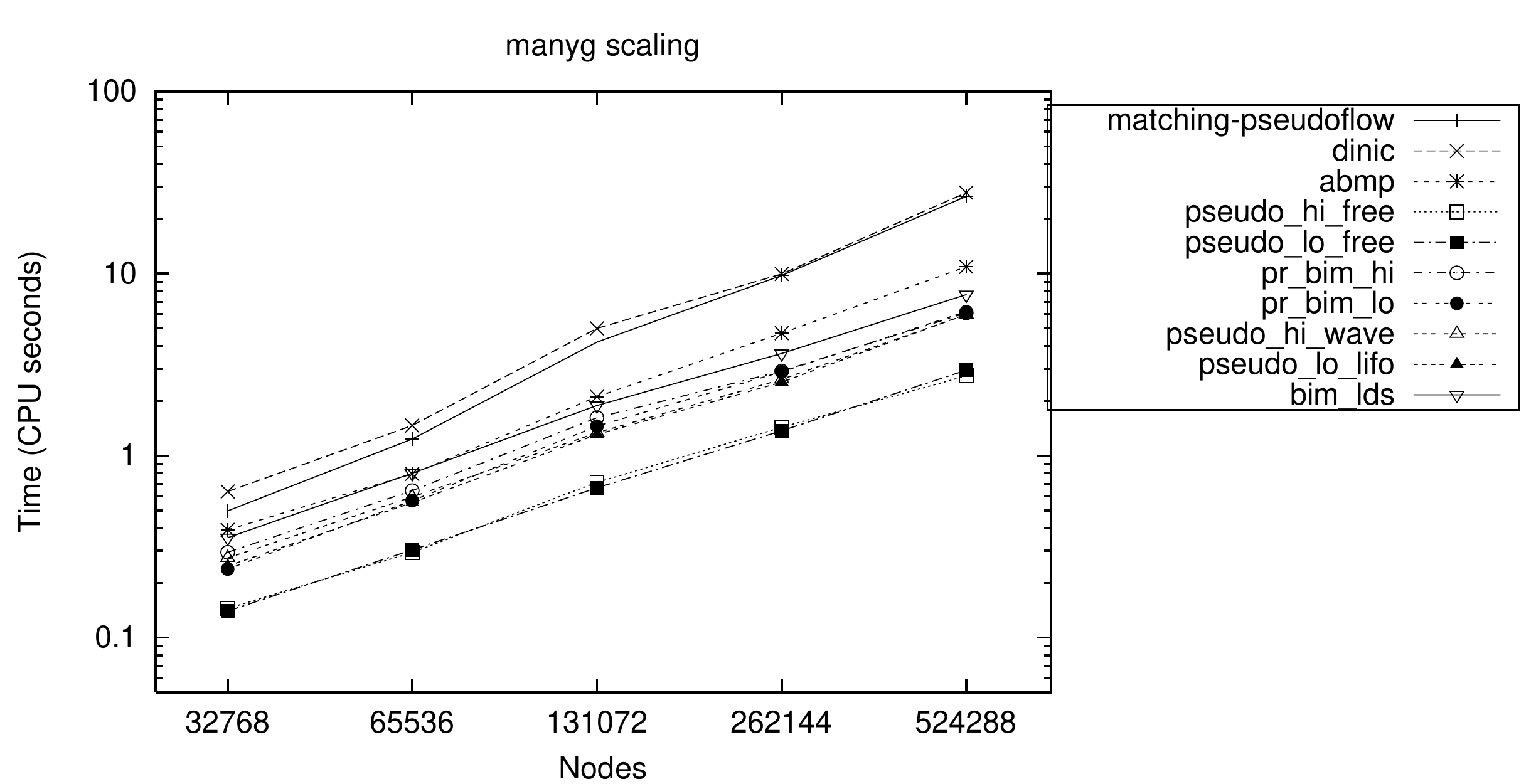}\bigskip
\begin{scriptsize}
\begin{tabular}{||c|r|r|r|r|r||}
\hline \hline
$n$, $m$ & \multicolumn{5}{|c||}{} \\ \hline
\multirow{6}{0.45in}{32,768 \\ 82,077}  &   \multicolumn{1}{|c|}{{\sf pseudo\_hi\_wave}}    &   \multicolumn{1}{|c|}{{\sf pseudo\_lo\_lifo}}    &   \multicolumn{1}{|c|}{{\sf pseudo\_hi\_free}}    &   \multicolumn{1}{|c|}{{\sf pseudo\_lo\_free}}    &   \multicolumn{1}{|c||}{{\sf matching-pseudoflow}}    \\  \cline{2-6}
    &   0.2742  &   0.2474  &   0.1448  &   0.1402  &   0.497   \\
    &   (1.956) &   (1.765) &   (1.033) &   ({\bf 1.000})   &   (3.545) \\  \cline{2-6}
    &   \multicolumn{1}{|c|}{{\sf abmp}}    &   \multicolumn{1}{|c|}{{\sf dinic}}   &   \multicolumn{1}{|c|}{{\sf pr\_bim\_hi}} &   \multicolumn{1}{|c|}{{\sf pr\_bim\_lo}} &   \multicolumn{1}{|c||}{{\sf bim\_lds}}   \\  \cline{2-6}
    &   0.3914  &   0.6356  &   0.2948  &   0.2386  &   0.354   \\
    &   (2.792) &   (4.534) &   (2.103) &   (1.702) &   (2.525) \\  \hline
\multirow{6}{0.45in}{65,536 \\ 163,719} &   \multicolumn{1}{|c|}{{\sf pseudo\_hi\_wave}}    &   \multicolumn{1}{|c|}{{\sf pseudo\_lo\_lifo}}    &   \multicolumn{1}{|c|}{{\sf pseudo\_hi\_free}}    &   \multicolumn{1}{|c|}{{\sf pseudo\_lo\_free}}    &   \multicolumn{1}{|c||}{{\sf matching-pseudoflow}}    \\
\hline
\hline
\end{tabular}
\end{scriptsize}
\caption{\label{Figure:manygscaling} Actual and relative run times for {\sf manyg} instances.}
\end{center}
\end{figure}

\begin{table}[ht]
\begin{center}
\begin{scriptsize}
\begin{tabular}{||c|l|r|r|r|r|r||}
\hline \hline
$n$, $m$ & \multicolumn{6}{|c||}{} \\ \hline
\multirow{10}{0.45in}{32,768 \\ 82,077} &       &   \multicolumn{1}{|c|}{{\sf pseudo\_hi\_wave}}    &   \multicolumn{1}{|c|}{{\sf pseudo\_lo\_lifo}}    &   \multicolumn{1}{|c|}{{\sf pseudo\_hi\_free}}    &   \multicolumn{1}{|c|}{{\sf pseudo\_lo\_free}}    &   \multicolumn{1}{|c||}{{\sf matching-pseudoflow}}    \\  \cline{2-7}
    &   Arc scans   &   192,549 &   184,781 &   248,068 &   232,279 &   \multicolumn{1}{|c||}{--}   \\
    &   Pushes  &   132,243 &   152,523 &   106,419 &   99,685  &   \multicolumn{1}{|c||}{--}   \\
    &   Relabels    &   91,643  &   88,978  &   44,199  &   41,694  &   \multicolumn{1}{|c||}{--}   \\
    &   Updates &   \multicolumn{1}{|c|}{--}    &   \multicolumn{1}{|c|}{--}    &   1   &   1   &   21  \\
    &   Mergers &   35,779  &   33,840  &   61,348  &   57,981  &   \multicolumn{1}{|c||}{--}   \\
    &   Depth   &   3.7 &   4.5 &   1.7 &   1.7 &   \multicolumn{1}{|c||}{--}   \\
\hline
\hline
\end{tabular}
\end{scriptsize}
\caption{\label{Table:manygopcount} Operation counts for {\sf manyg} instances.}
\end{center}
\end{table}

\clearpage
\newpage

\begin{figure}[ht]
\begin{center}
\includegraphics[width = \linewidth]{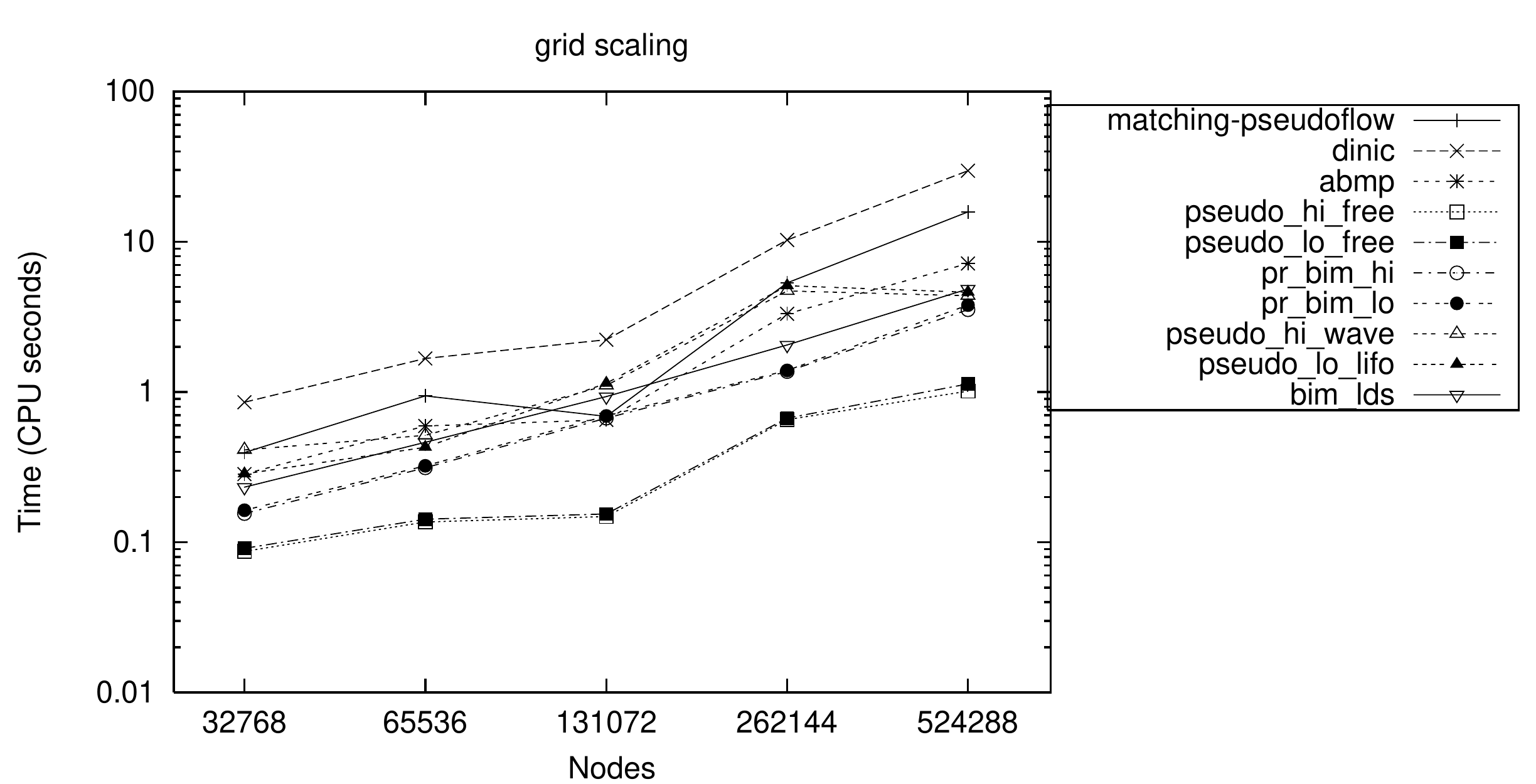}\bigskip
\begin{scriptsize}
\begin{tabular}{||c|r|r|r|r|r||}
\hline \hline
$n$, $m$ & \multicolumn{5}{|c||}{} \\ \hline
\multirow{6}{0.45in}{32,768 \\ 98,304}  &   \multicolumn{1}{|c|}{{\sf pseudo\_hi\_wave}}    &   \multicolumn{1}{|c|}{{\sf pseudo\_lo\_lifo}}    &   \multicolumn{1}{|c|}{{\sf pseudo\_hi\_free}}    &   \multicolumn{1}{|c|}{{\sf pseudo\_lo\_free}}    &   \multicolumn{1}{|c||}{{\sf matching-pseudoflow}}    \\  \cline{2-6}
    &   0.412   &   0.286   &   0.087   &   0.091   &   0.396   \\
    &   (4.736) &   (3.283) &   ({\bf 1.000})   &   (1.046) &   (4.552) \\  \cline{2-6}
    &   \multicolumn{1}{|c|}{{\sf abmp}}    &   \multicolumn{1}{|c|}{{\sf dinic}}   &   \multicolumn{1}{|c|}{{\sf pr\_bim\_hi}} &   \multicolumn{1}{|c|}{{\sf pr\_bim\_lo}} &   \multicolumn{1}{|c||}{{\sf bim\_lds}}   \\  \cline{2-6}
    &   0.284   &   0.856   &   0.155   &   0.163   &   0.233   \\
    &   (3.269) &   (9.839) &   (1.784) &   (1.878) &   (2.676) \\  \hline
\multirow{6}{0.45in}{65,536 \\ 196,608} &   \multicolumn{1}{|c|}{{\sf pseudo\_hi\_wave}}    &   \multicolumn{1}{|c|}{{\sf pseudo\_lo\_lifo}}    &   \multicolumn{1}{|c|}{{\sf pseudo\_hi\_free}}    &   \multicolumn{1}{|c|}{{\sf pseudo\_lo\_free}}    &   \multicolumn{1}{|c||}{{\sf matching-pseudoflow}}    \\
\hline
\hline
\end{tabular}
\end{scriptsize}
\caption{\label{Figure:gridscaling} Actual and relative run times for {\sf grid} instances.}
\end{center}
\end{figure}

\begin{table}[ht]
\begin{center}
\begin{scriptsize}
\begin{tabular}{||c|l|r|r|r|r|r||}
\hline \hline
$n$, $m$ & \multicolumn{6}{|c||}{} \\ \hline
\multirow{10}{0.45in}{32,768 \\ 98,304} &       &   \multicolumn{1}{|c|}{{\sf pseudo\_hi\_wave}}    &   \multicolumn{1}{|c|}{{\sf pseudo\_lo\_lifo}}    &   \multicolumn{1}{|c|}{{\sf pseudo\_hi\_free}}    &   \multicolumn{1}{|c|}{{\sf pseudo\_lo\_free}}    &   \multicolumn{1}{|c||}{{\sf matching-pseudoflow}}    \\  \cline{2-7}
    &   Arc scans   &   132,434 &   131,001 &   154,177 &   153,498 &   \multicolumn{1}{|c||}{--}   \\
    &   Pushes  &   221,811 &   172,938 &   58,471  &   57,805  &   \multicolumn{1}{|c||}{--}   \\
    &   Relabels    &   38,819  &   38,902  &   20,801  &   20,689  &   \multicolumn{1}{|c||}{--}   \\
    &   Updates &   \multicolumn{1}{|c|}{--}    &   \multicolumn{1}{|c|}{--}    &   0   &   0   &   10  \\
    &   Mergers &   32,806  &   32,075  &   37,427  &   37,094  &   \multicolumn{1}{|c||}{--}   \\
    &   Depth   &   6.8 &   5.4 &   1.6 &   1.6 &   \multicolumn{1}{|c||}{--}   \\
\hline
\hline
\end{tabular}
\end{scriptsize}
\caption{\label{Table:gridopcount} Operation counts for {\sf grid} instances.}
\end{center}
\end{table}

\clearpage
\newpage

\begin{figure}[ht]
\begin{center}
\includegraphics[width = \linewidth]{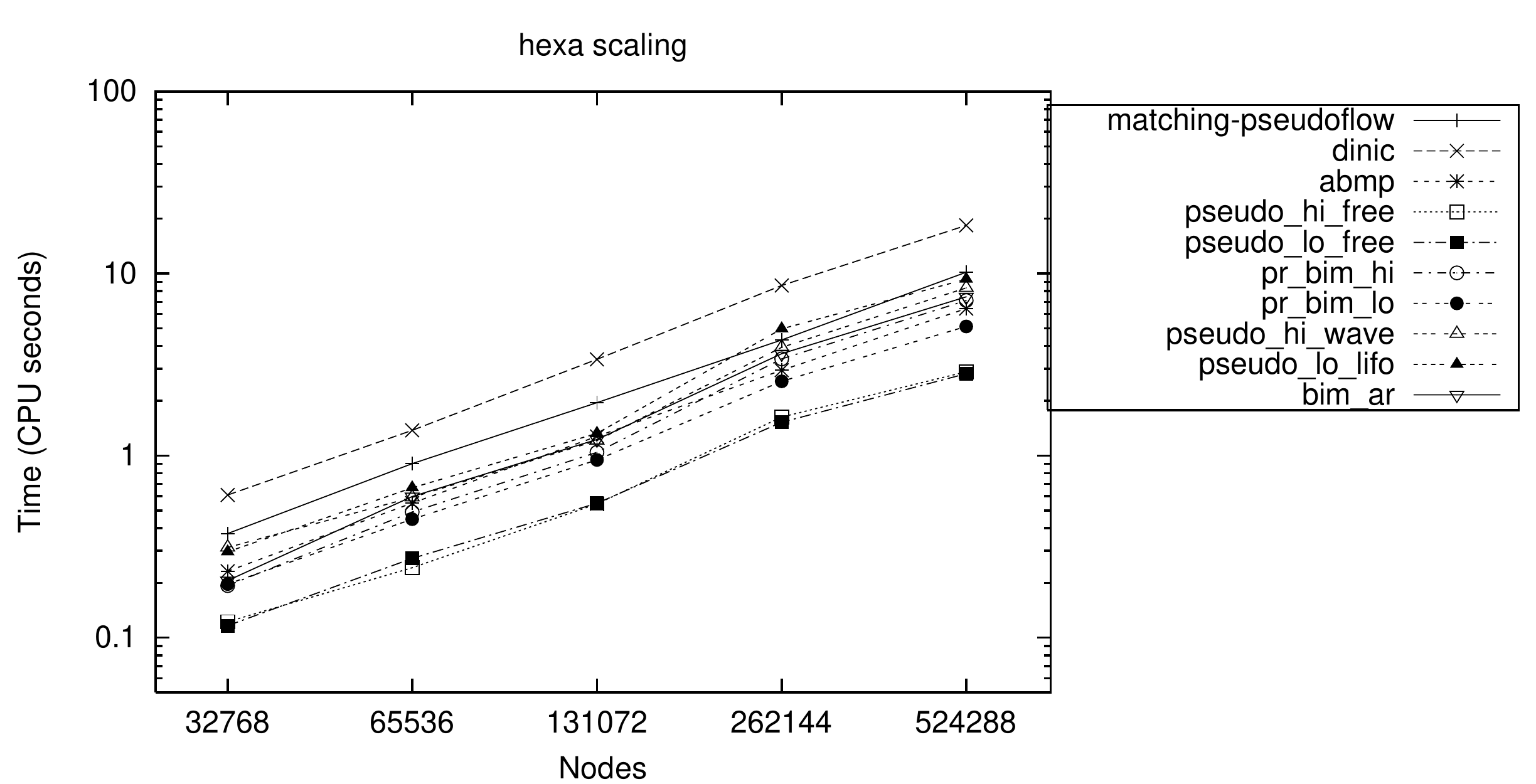}\bigskip
\begin{scriptsize}
\begin{tabular}{||c|r|r|r|r|r||}
\hline \hline
$n$, $m$ & \multicolumn{5}{|c||}{} \\ \hline
\multirow{6}{0.45in}{32,768 \\ 98,304}  &   \multicolumn{1}{|c|}{{\sf pseudo\_hi\_wave}}    &   \multicolumn{1}{|c|}{{\sf pseudo\_lo\_lifo}}    &   \multicolumn{1}{|c|}{{\sf pseudo\_hi\_free}}    &   \multicolumn{1}{|c|}{{\sf pseudo\_lo\_free}}    &   \multicolumn{1}{|c||}{{\sf matching-pseudoflow}}    \\  \cline{2-6}
    &   0.313   &   0.297   &   0.122   &   0.117   &   0.373   \\
    &   (2.686) &   (2.544) &   (1.050) &   ({\bf 1.000})   &   (3.199) \\  \cline{2-6}
    &   \multicolumn{1}{|c|}{{\sf abmp}}    &   \multicolumn{1}{|c|}{{\sf dinic}}   &   \multicolumn{1}{|c|}{{\sf pr\_bim\_hi}} &   \multicolumn{1}{|c|}{{\sf pr\_bim\_lo}} &   \multicolumn{1}{|c||}{{\sf bim\_ar}}    \\  \cline{2-6}
    &   0.232   &   0.608   &   0.193   &   0.197   &   0.206   \\
    &   (1.986) &   (5.218) &   (1.657) &   (1.691) &   (1.767) \\  \hline
\multirow{6}{0.45in}{65,536 \\ 196,608} &   \multicolumn{1}{|c|}{{\sf pseudo\_hi\_wave}}    &   \multicolumn{1}{|c|}{{\sf pseudo\_lo\_lifo}}    &   \multicolumn{1}{|c|}{{\sf pseudo\_hi\_free}}    &   \multicolumn{1}{|c|}{{\sf pseudo\_lo\_free}}    &   \multicolumn{1}{|c||}{{\sf matching-pseudoflow}}    \\
\hline
\hline
\end{tabular}
\end{scriptsize}
\caption{\label{Figure:hexascaling} Actual and relative run times for {\sf hexa} instances.}
\end{center}
\end{figure}

\begin{table}[ht]
\begin{center}
\begin{scriptsize}
\begin{tabular}{||c|l|r|r|r|r|r||}
\hline \hline
$n$, $m$ & \multicolumn{6}{|c||}{} \\ \hline
\multirow{10}{0.45in}{32,768 \\ 98,304} &       &   \multicolumn{1}{|c|}{{\sf pseudo\_hi\_wave}}    &   \multicolumn{1}{|c|}{{\sf pseudo\_lo\_lifo}}    &   \multicolumn{1}{|c|}{{\sf pseudo\_hi\_free}}    &   \multicolumn{1}{|c|}{{\sf pseudo\_lo\_free}}    &   \multicolumn{1}{|c||}{{\sf matching-pseudoflow}}    \\  \cline{2-7}
    &   Arc scans   &   248,522 &   218,196 &   225,952 &   218,890 &   \multicolumn{1}{|c||}{--}   \\
    &   Pushes  &   170,434 &   245,992 &   79,222  &   77,716  &   \multicolumn{1}{|c||}{--}   \\
    &   Relabels    &   92,615  &   79,810  &   30,985  &   31,041  &   \multicolumn{1}{|c||}{--}   \\
    &   Updates &   \multicolumn{1}{|c|}{--}    &   \multicolumn{1}{|c|}{--}    &   1   &   1   &   8   \\
    &   Mergers &   33,431  &   31,771  &   47,412  &   46,659  &   \multicolumn{1}{|c||}{--}   \\
    &   Depth   &   5.1 &   7.7 &   1.7 &   1.7 &   \multicolumn{1}{|c||}{--}   \\
\hline
\hline
\end{tabular}
\end{scriptsize}
\caption{\label{Table:hexaopcount}Operation counts for {\sf hexa} instances.}
\end{center}
\end{table}

\clearpage
\newpage

\begin{figure}[ht]
\begin{center}
\includegraphics[width = \linewidth]{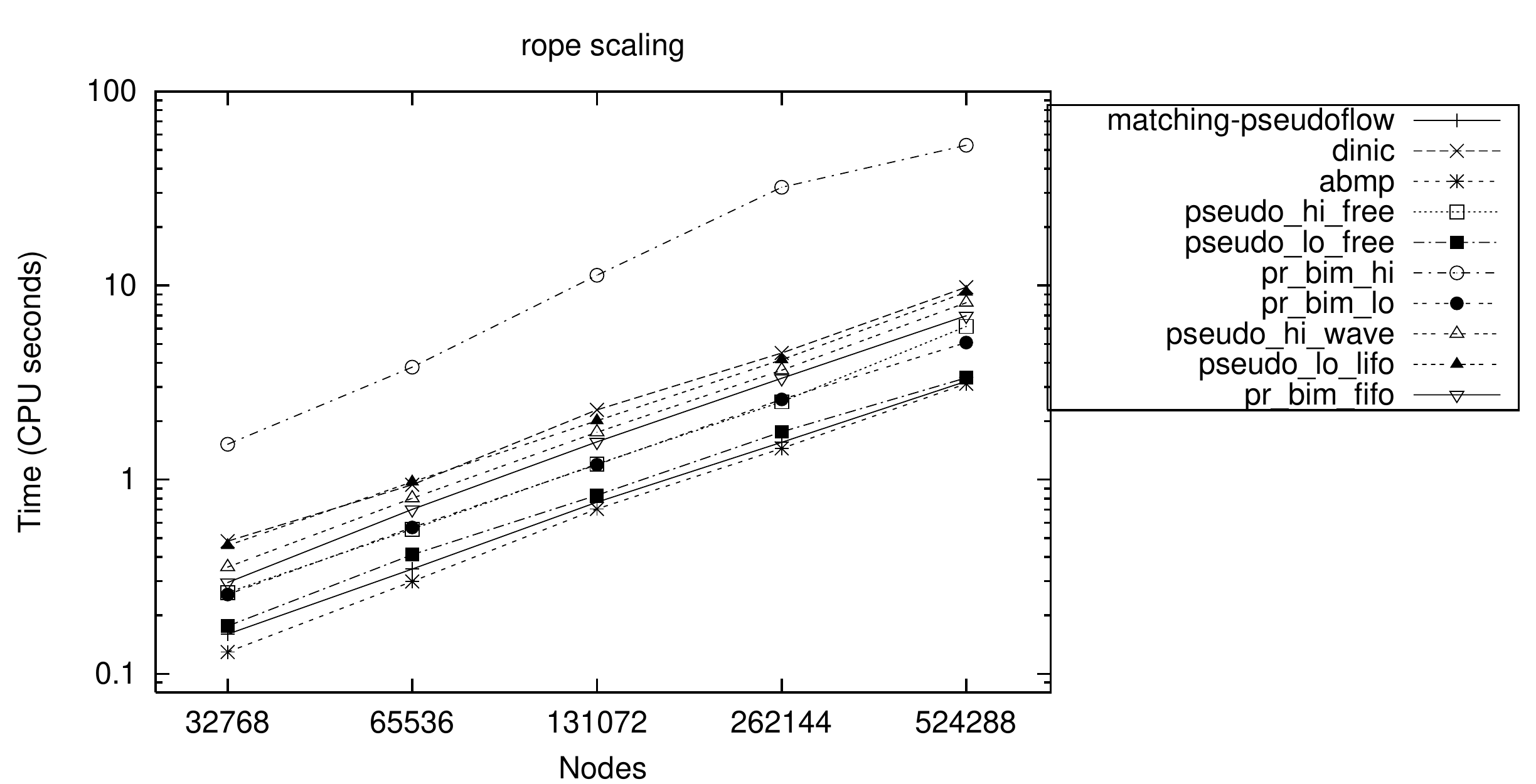}\bigskip
\begin{scriptsize}
\begin{tabular}{||c|r|r|r|r|r||}
\hline \hline
$n$, $m$ & \multicolumn{5}{|c||}{} \\ \hline
\multirow{6}{0.45in}{32,768 \\ 98,304}  &   \multicolumn{1}{|c|}{{\sf pseudo\_hi\_wave}}    &   \multicolumn{1}{|c|}{{\sf pseudo\_lo\_lifo}}    &   \multicolumn{1}{|c|}{{\sf pseudo\_hi\_free}}    &   \multicolumn{1}{|c|}{{\sf pseudo\_lo\_free}}    &   \multicolumn{1}{|c||}{{\sf matching-pseudoflow}}    \\  \cline{2-6}
    &   0.354   &   0.457   &   0.262   &   0.176   &   0.160   \\
    &   (2.743) &   (3.540) &   (2.026) &   (1.361) &   (1.237) \\  \cline{2-6}
    &   \multicolumn{1}{|c|}{{\sf abmp}}    &   \multicolumn{1}{|c|}{{\sf dinic}}   &   \multicolumn{1}{|c|}{{\sf pr\_bim\_hi}} &   \multicolumn{1}{|c|}{{\sf pr\_bim\_lo}} &   \multicolumn{1}{|c||}{{\sf pr\_bim\_fifo}}  \\  \cline{2-6}
    &   0.129   &   0.482   &   1.521   &   0.256   &   0.295   \\
    &   ({\bf 1.000})   &   (3.729) &   (11.769)    &   (1.981) &   (2.280) \\  \hline
\multirow{6}{0.45in}{65,536 \\ 196,608} &   \multicolumn{1}{|c|}{{\sf pseudo\_hi\_wave}}    &   \multicolumn{1}{|c|}{{\sf pseudo\_lo\_lifo}}    &   \multicolumn{1}{|c|}{{\sf pseudo\_hi\_free}}    &   \multicolumn{1}{|c|}{{\sf pseudo\_lo\_free}}    &   \multicolumn{1}{|c||}{{\sf matching-pseudoflow}}    \\
\hline
\hline
\end{tabular}
\end{scriptsize}
\caption{\label{Figure:ropescaling}Actual and relative run times for {\sf rope} instances.}
\end{center}
\end{figure}

\begin{table}[ht]
\begin{center}
\begin{scriptsize}
\begin{tabular}{||c|l|r|r|r|r|r||}
\hline \hline
$n$, $m$ & \multicolumn{6}{|c||}{} \\ \hline
\multirow{10}{0.45in}{32,768 \\ 98,304} &       &   \multicolumn{1}{|c|}{{\sf pseudo\_hi\_wave}}    &   \multicolumn{1}{|c|}{{\sf pseudo\_lo\_lifo}}    &   \multicolumn{1}{|c|}{{\sf pseudo\_hi\_free}}    &   \multicolumn{1}{|c|}{{\sf pseudo\_lo\_free}}    &   \multicolumn{1}{|c||}{{\sf matching-pseudoflow}}    \\  \cline{2-7}
    &   Arc scans   &   450,472 &   640,668 &   284,307 &   185,589 &   \multicolumn{1}{|c||}{--}   \\
    &   Pushes  &   109,632 &   177,015 &   125,871 &   67,952  &   \multicolumn{1}{|c||}{--}   \\
    &   Relabels    &   184,689 &   266,637 &   33,881  &   16,449  &   \multicolumn{1}{|c||}{--}   \\
    &   Updates &   \multicolumn{1}{|c|}{--}    &   \multicolumn{1}{|c|}{--}    &   4   &   1   &   12  \\
    &   Mergers &   50,793  &   59,201  &   71,126  &   42,166  &   \multicolumn{1}{|c||}{--}   \\
    &   Depth   &   2.2 &   3.0 &   1.8 &   1.6 &   \multicolumn{1}{|c||}{--}   \\
\hline
\hline
\end{tabular}
\end{scriptsize}
\caption{\label{Table:ropeopcount} Operation counts for {\sf rope} instances.}
\end{center}
\end{table}

\clearpage
\newpage

\begin{figure}[ht]
\begin{center}
\includegraphics[width = \linewidth]{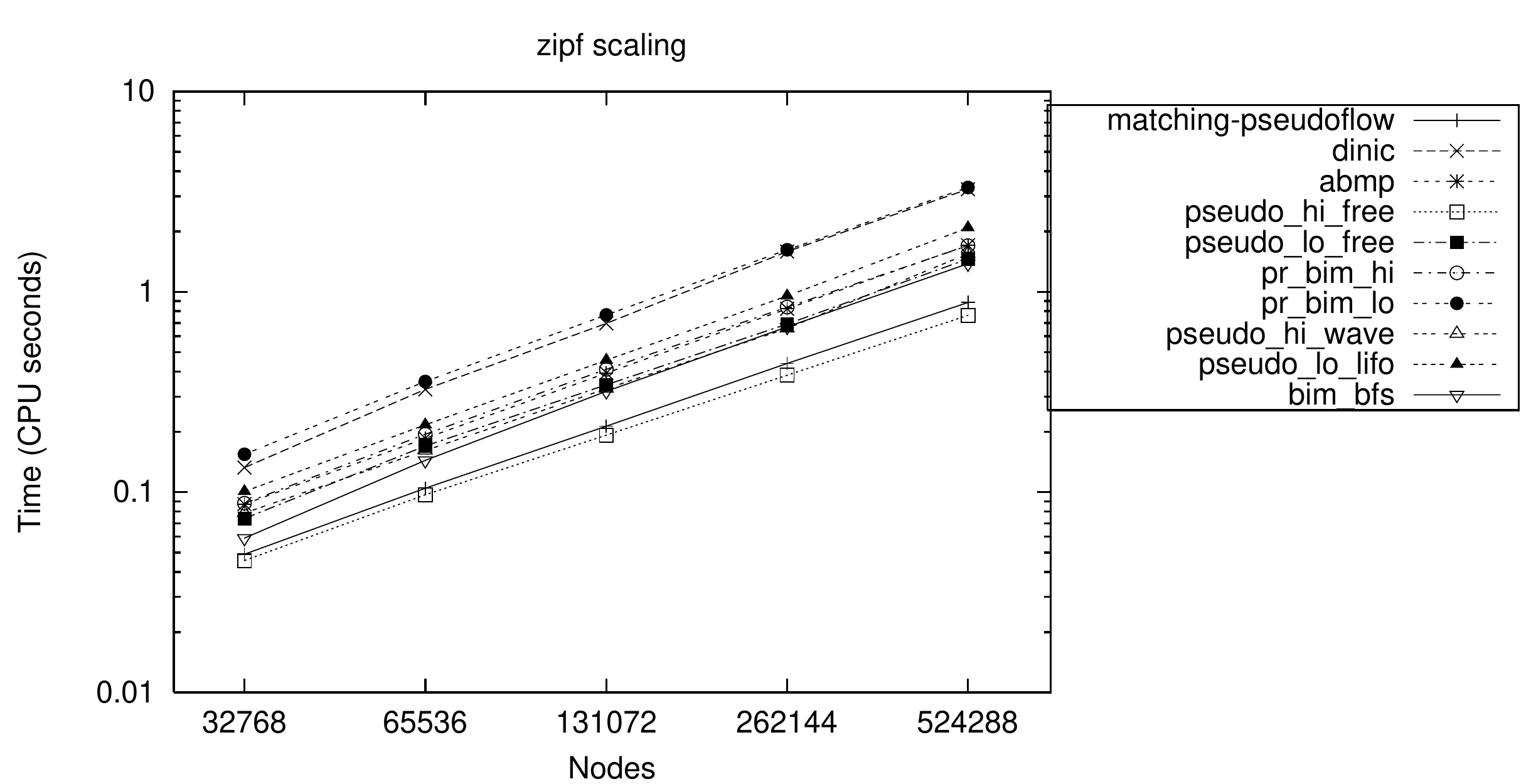}\bigskip
\begin{scriptsize}
\begin{tabular}{||c|r|r|r|r|r||}
\hline \hline
$n$, $m$ & \multicolumn{5}{|c||}{} \\ \hline
\multirow{6}{0.45in}{32,768 \\ 98,304}  &   \multicolumn{1}{|c|}{{\sf pseudo\_hi\_wave}}    &   \multicolumn{1}{|c|}{{\sf pseudo\_lo\_lifo}}    &   \multicolumn{1}{|c|}{{\sf pseudo\_hi\_free}}    &   \multicolumn{1}{|c|}{{\sf pseudo\_lo\_free}}    &   \multicolumn{1}{|c||}{{\sf matching-pseudoflow}}    \\  \cline{2-6}
    &   0.079   &   0.101   &   0.046   &   0.074   &   0.049   \\
    &   (1.724) &   (2.211) &   ({\bf 1.000})   &   (1.618) &   (1.075) \\  \cline{2-6}
    &   \multicolumn{1}{|c|}{{\sf abmp}}    &   \multicolumn{1}{|c|}{{\sf dinic}}   &   \multicolumn{1}{|c|}{{\sf pr\_bim\_hi}} &   \multicolumn{1}{|c|}{{\sf pr\_bim\_lo}} &   \multicolumn{1}{|c||}{{\sf bim\_bfs}}   \\  \cline{2-6}
    &   0.087   &   0.133   &   0.088   &   0.154   &   0.059   \\
    &   (1.912) &   (2.908) &   (1.930) &   (3.386) &   (1.294) \\  \hline
\multirow{6}{0.45in}{65,536 \\ 196,608} &   \multicolumn{1}{|c|}{{\sf pseudo\_hi\_wave}}    &   \multicolumn{1}{|c|}{{\sf pseudo\_lo\_lifo}}    &   \multicolumn{1}{|c|}{{\sf pseudo\_hi\_free}}    &   \multicolumn{1}{|c|}{{\sf pseudo\_lo\_free}}    &   \multicolumn{1}{|c||}{{\sf matching-pseudoflow}}    \\
\hline
\hline
\end{tabular}
\end{scriptsize}
\caption{\label{Figure:zipfscaling} Actual and relative run times for {\sf zipf} instances.}
\end{center}
\end{figure}

\begin{table}[ht]
\begin{center}
\begin{scriptsize}
\begin{tabular}{||c|l|r|r|r|r|r||}
\hline \hline
$n$, $m$ & \multicolumn{6}{|c||}{} \\ \hline
\multirow{10}{0.45in}{32,768 \\ 98,304} &       &   \multicolumn{1}{|c|}{{\sf pseudo\_hi\_wave}}    &   \multicolumn{1}{|c|}{{\sf pseudo\_lo\_lifo}}    &   \multicolumn{1}{|c|}{{\sf pseudo\_hi\_free}}    &   \multicolumn{1}{|c|}{{\sf pseudo\_lo\_free}}    &   \multicolumn{1}{|c||}{{\sf matching-pseudoflow}}    \\  \cline{2-7}
    &   Arc scans   &   72,958  &   72,086  &   93,511  &   72,960  &   \multicolumn{1}{|c||}{--}   \\
    &   Pushes  &   21,603  &   20,103  &   31,891  &   22,717  &   \multicolumn{1}{|c||}{--}   \\
    &   Relabels    &   48,253  &   48,278  &   23,393  &   17,549  &   \multicolumn{1}{|c||}{--}   \\
    &   Updates &   \multicolumn{1}{|c|}{--}    &   \multicolumn{1}{|c|}{--}    &   0   &   1   &   5   \\
    &   Mergers &   19,124  &   17,866  &   23,446  &   18,859  &   \multicolumn{1}{|c||}{--}   \\
    &   Depth   &   1.1 &   1.1 &   1.4 &   1.2 &   \multicolumn{1}{|c||}{--}   \\
\hline
\hline
\end{tabular}
\end{scriptsize}
\caption{\label{Table:zipfopcount} Operation counts for {\sf zipf} instances.}
\end{center}
\end{table}

\clearpage
\newpage
\appendix

\section{Proof of Lemma \ref{lem:stages}}
\label{proof:stages}

\begin{proof}
Our analysis of the number of stages is essentially the same as that of Dinic's algorithm as per Even and Tarjan \cite{EveT75} and Hopcroft and Karp \cite{HopK73}.

By construction, each $\ell$ layered network guarantees at least one successful path, as some $WT_1$ node is reachable through a sequence of mergers of length $\ell$. We divide the stages into two parts: the first part includes stages of labels no larger than $\sqrt{\kappa}$, and the second part consists of the stages with labels greater than $\sqrt{\kappa}$.  Since the label of the lowest labeled strong node strictly increases in each stage, the number of stages in the first part is at most $\sqrt{\kappa}$. We show that the second part can also have at most $\sqrt{\kappa}$ stages.

In the second part of the algorithm, the successful paths of length $L > \sqrt{\kappa}$ are equivalent to flow augmentations along a path of length $2L+1$.  We now observe that each $WT_2$ branch contains a residual arc of capacity $1$ from root to child, and the set of residual arcs in the $WT_2$ branches in any layer $p> 0$ of the network forms a valid cut in the residual graph.  This is since it separates the roots in this layer and nodes with label greater than $p$ from the children in the layer and nodes with label less than $p$ as in Figure \ref{fig:cut}.

\begin{figure}[ht]
\centerline{\includegraphics[width=0.85\linewidth]{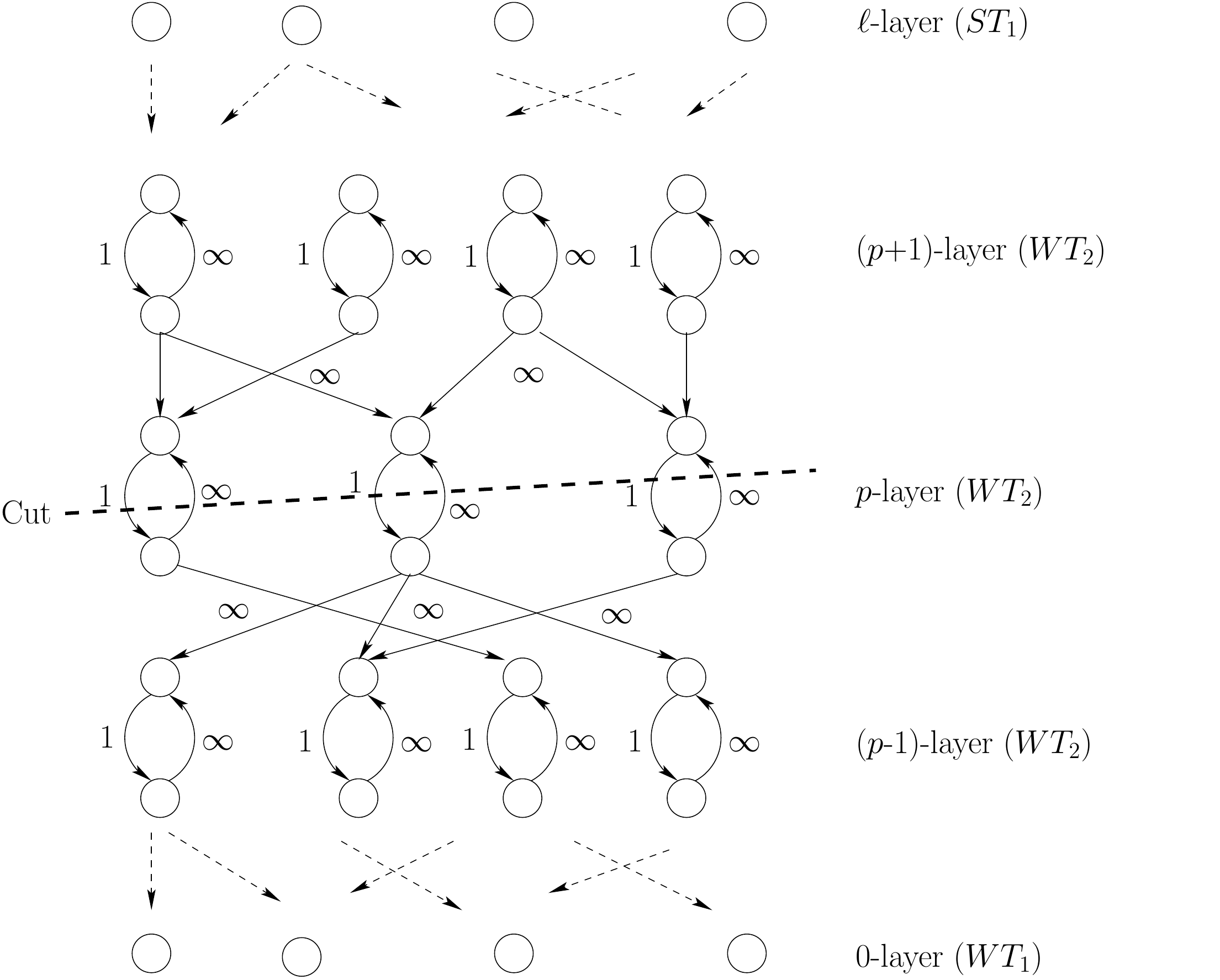}}
\caption{\label{fig:cut}Arcs in the branches of a layer form a cut in the residual graph.}
\end{figure}

Thus the maximum flow value in the residual graph at the beginning of part two is no larger than the smallest number of branches in one of the layers. Since the layered network consists of at most $\kappa$ branches and the number of layers is $L$, then the maximum flow in the residual graph can be no larger than $\frac{\kappa }{L}$, which in the second part is no larger than $\sqrt{\kappa}$. Thus the total number of augmentations in part two is at most $\sqrt{\kappa}$. Since each layered network guarantees at least once augmentation, there are at most $\sqrt{\kappa}$ stages in the second part of the algorithm. \qed
\end{proof}

\section{Complexity using word operations}
\label{sec:words}
We show here how to use boolean operations to improve the complexity of the {\sf matching-pseudoflow} algorithm. The characteristic vector of out-neighbors of each $V_1$-node $v$ is maintained as a binary word OUT($v$) of length $n_2$. OUT($v$) is a word where the $i^{th}$ bit is 1 if there is an arc from $v \in V_1$ to $i \in V_2$. We also maintain a characteristic vector of in-neighbors list as a word IN($v$) of length $n_1$ for each node in $v\in V_2$. IN($v$) is a word where the $i^{th}$ bit is 1 if the arc from $i \in V_1$ to $v \in V_2$ exists.

These words are maintained {\em in addition to the adjacency list} which is a linked list of in and out neighbors for each node in the graph. The words and the adjacency list are used in parallel to achieve the better time complexity of the approach using only words and that using only the adjacency list. When we say that the two are used in parallel we imply that the adjacency list and word operations are accessed and used alternately.

Using $\lambda$-bit word operations ($\lambda < n_1$), we break OUT() and IN() into a concatenation of $\lambda$-bit words, and perform operations on these words. Each of these $\lambda$-bit words is called a $\lambda$-{\emph word} and the $j^{th}$ $\lambda$-word is denoted by OUT$^{j}$($v$).

Three boolean operations are used:
\begin{enumerate}
\item LEAD:  Given a word $W$, {\tt lead}($W$) returns the index of the leading non-zero bit in $W$, and 0 if all bits are 0.
\item AND: Given two words $A$ and $B$ of the same length, $A \wedge B$ is a word whose $i^{th}$ bit is 1 iff the $i^{th}$ bits of $A$ and $B$ are 1, and 0 otherwise.
\item OR: Given two words $A$ and $B$ of the same length, $A \vee B$ is a word whose $i^{th}$ bit is 1 if the $i^{th}$ bit of $A$ or $B$ (or both) is 1, and 0 otherwise.
\end{enumerate}

If we wish to perform any of the above operations on a word of $k$ bits using word operations on words of $\lambda$ bits where $\lambda < k$, each $k$-bit word operation can be done in $O(\frac{k}{\lambda})$ steps.

Any boolean operation ($\wedge$, $\vee$, {\tt lead}) on a $\lambda$-word counts as a single operation. Given two nodes $i \in V_1$ and $j \in V_2$, the bits corresponding to the arc $(i,j)$ in IN($j$) and OUT($i$) can be accessed and modified in $O(1)$.

{\bf Initialization~} For each node $v \in V_1$, we look at the next
arc in its out neighbors in the adjacency list. If this arc does not
lead to an unmatched $WT_1$ node, we perform a {\tt
lead}(OUT$^{1}$($v$)) operation. If {\tt lead}(OUT$^{1}$($v$)) equals
0, we return to the adjacency list and look at the next arc. Again, if
this arc does not lead to an unmatched $V_2$-node, a {\tt lead}
operation is performed on the next unscanned $\lambda$-word
(OUT$^{2}$($v$)). This procedure of looking at the next $\lambda$-word
and the next arc in the adjacency list until an unmatched $V_2$
neighbor is found, or the end of the list is reached. In the adjacency
list, either a neighboring $WT_1$ branch is found or all the neighbors
are exhausted in at most $\kappa$ arc scans for each $V_1$-node. Thus,
there are at most $n_1 \kappa$ arc scans. Further, each arc is looked
at most once, so the complexity is $O(\min\{n_1\kappa, m\})$.

In OUT($v$), either a neighboring $WT_1$ branch is found or all the neighbors are exhausted in $O(n_2/\lambda)$ operations. A neighboring $WT_1$ branch, if it exists, is thus found in $O(\min\{n_1\kappa, m, \frac{n_1n_2}{\lambda}\})$.

Once a merger is executed, the bits corresponding to the merger arc in the IN() and OUT() words must be changed. Since there are at most $\kappa$ mergers during initialization, and each requires $O(1)$ work, the work done to maintain these words is $O(\kappa)$.

\begin{claim}
The work done in initialization using $\lambda$-words is $O(\min\{n_1\kappa, m, \frac{n_1n_2}{\lambda}\})$.
\end{claim}

{\bf Building the 1-layer~} Similar to the initialization, for each child $v$ of a $WT_2$ branch, we search for a merger arc by looking in parallel at the next neighbor in the adjacency list and performing a {\tt lead}() operation on the next $\lambda$-word OUT($v$). The search terminates either when a $WT_1$ neighbor is found or the end of the list is reached, which occurs in $\min\{\kappa, n_2/\lambda\}$ operations. Since there are at most $kappa$ nodes that are children of a $WT_2$ branch and each arc is scanned at most once in the entire algorithm, the total work to generate the $1$-layer of the layered network throughout the algorithm is $O(\min\{\kappa^2, n_2\kappa/\lambda, m\})$.

{\bf Building the layered network~} We now describe the use of word operations in generating a layered network upwards from the 1-layer. With the exception of the $\ell$-layer, all the branches in the layered network are $WT_2$ branches. We first discuss labeling the $WT_2$ branches, and later discuss how to find the $\ell$-layer.

We construct words for the sub-graph induced only by the nodes in the $WT_2$ branches. That is, each node $v \in V_1$ which is the child of a $WT_2$ branch has an associated word SUB-OUT($v$) containing the subset of its out-neighbor nodes that are in $WT_2$ branches, the length of which is at most $\kappa$. Note that this is different from OUT($v$) which is a word of length $n_2$ and contains all out neighbors of $v$, not just those that are in $WT_2$ branches.

The $i^{th}$ bit of SUB-OUT($v$) is 1 if an arc exists from $v$ to the $i^{th}$ node which is a root of a $WT_2$ branch. Similarly, the roots of the $WT_2$ branches have a word SUB-IN($v$) of size at most $\kappa$ representing the in-neighbors of $v$ that are children in a $WT_2$ branch.  We will use SUB-IN() to build the layered network and SUB-OUT() while pushing flow through this network.

Initially, SUB-IN() and SUB-OUT() are empty since there are no $WT_2$ branches.  As $WT_2$ branches are created during the algorithm, SUB-IN() SUB-OUT() words are created for each of the nodes in these branches. At any point in the algorithm, SUB-OUT($v$) is a subset of OUT($v$) and SUB-IN($v$) is a subset of IN($v$) that contains only those bits that correspond to nodes that are in $WT_2$ branches. The relation between IN(), OUT(), SUB-IN() and SUB-OUT() are shown in Figure \ref{fig:words1}. The matrix formed by the SUB-IN() and SUB-OUT() words is referred to as the SUB-matrix.

\begin{figure}[ht]
\centerline{\includegraphics[width=0.8\textwidth]{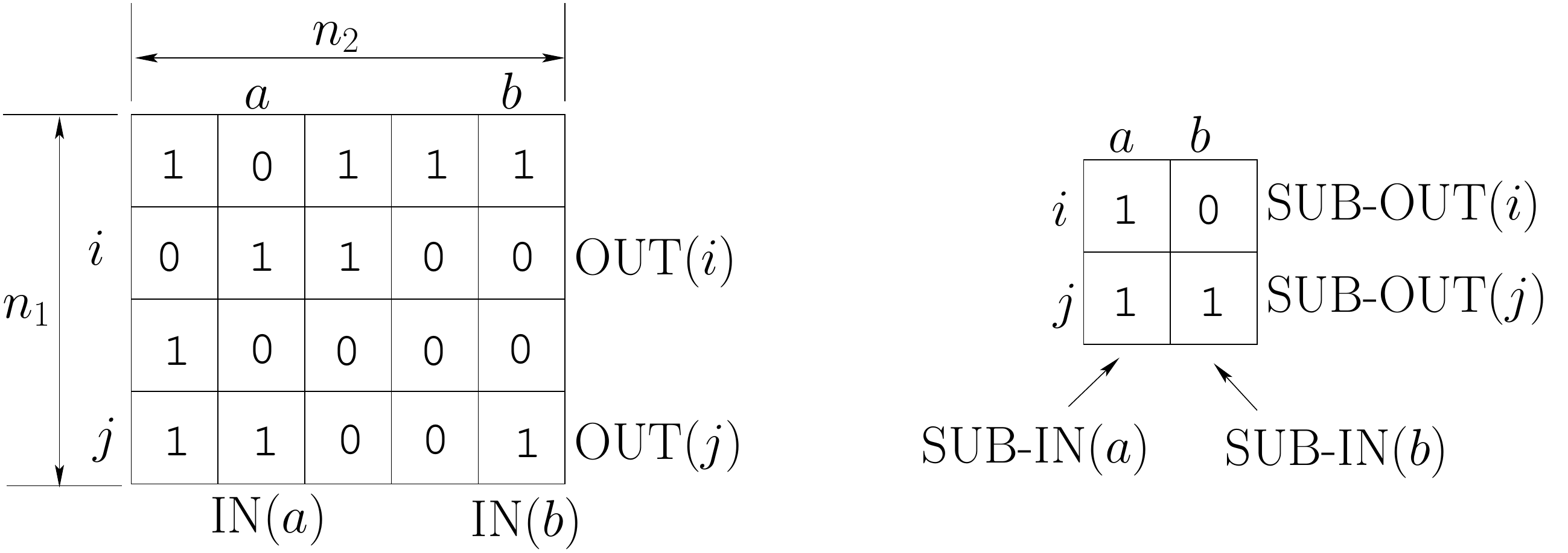}}
\caption{\label{fig:words1}IN(), OUT(), SUB-IN() and SUB-OUT() at some point during the algorithm when nodes $i, j \in V_1$, and $a, b \in V_2$ are in $WT_2$ branches.}
\end{figure}

Each time a singleton node (either $WT_1$ or $ST_1$ branch) becomes part of a $WT_2$ branch either during initialization or later in the algorithm, we add a bit corresponding to that node to the existing SUB-IN() and SUB-OUT() words, and create a new word for that node. This is equivalent to adding a row and column to the SUB-matrix. The total work  throughout the algorithm is $O(\kappa^2)$ since there are $O(\kappa)$ words each of size $O(\kappa)$, and adding a bit to the words is an $O(1)$ operation.

Two more types of words are needed to construct the layered network.
\begin{enumerate}
\item A $V_2$LAYER($k$) word (the characteristic vector of each layer) indicating the $V_2$-nodes contained in each layer $1 \leq k \leq \kappa$.  The length of the word is at most $\kappa$ and a bit of $V_2$LAYER($k$) is 1 if a $V_2$-node corresponding to that bit is in layer $k$. All the $V_2$LAYER() words are set to 0 at the beginning of each stage.
\item A $V_1$LAYER($k$) word (the characteristic vector of each layer) indicating the $V_1$-nodes contained in each layer $1 \leq k \leq \kappa$. The length of the word is at most $\kappa$ and a bit of $V_1$LAYER($k$) is 1 if a $V_1$-node corresponding to that bit is in layer $k$. All the $V_1$LAYER() words are set to 0 at the beginning of each stage.
\item A REACHED word of size $\kappa$ that keeps track of $V_1$-nodes that have {\em not} been reached by the upward breadth-first-search in each stage to create the layered network. The $i^{th}$ bit is 0 if that node has been assigned to a layer, and 1 otherwise. All the bits of this word are set to 1 at the beginning of each stage.
\end{enumerate}

Once we have computed the 1-layer of the network, $V_1$LAYER(1) is populated with 1 in the locations of the nodes in the 1-layer. REACHED is then populated with 0 in the locations of the $V_1$-nodes in the 1-layer.

The $V_2$LAYER(1) is now constructed by successively looking at each $\lambda$-word in $V_1$LAYER(1) and performing a {\tt lead}() operation on that word. If the result of the {\tt lead}() operation is non-zero, then we know the index of a $V_1$-node in the 1-layer, and its unique parent's bit is changed in the $V_2$LAYER(1). The 1-bit corresponding to the output of the {\tt lead}() operation is now set to 0, and another {\tt lead}() operation is performed on the same word.  Identifying a 1-bit and changing it continues until the result of the {\tt lead}() operation is zero, in which case we move to the next $\lambda$-word and perform a {\tt lead}() operation on that word to find a 1-bit.

Now, given $V_2$-nodes $\{v_1, \ldots, v_j\}$ in the 1-layer, $V_1$LAYER(2) is
\[
V_1\mathrm{LAYER(2)} = \mathrm{(SUB-IN(}v_1\mathrm{)} \vee \mathrm{SUB-IN(}v_2\mathrm{)} \vee \ldots \vee
\mathrm{SUB-IN(}v_j\mathrm{))}\wedge \mathrm{REACHED}.
\]
Figure \ref{fig:words2} illustrates this procedure on an example.

\begin{figure}[t]
\centerline{\includegraphics[width=\textwidth]{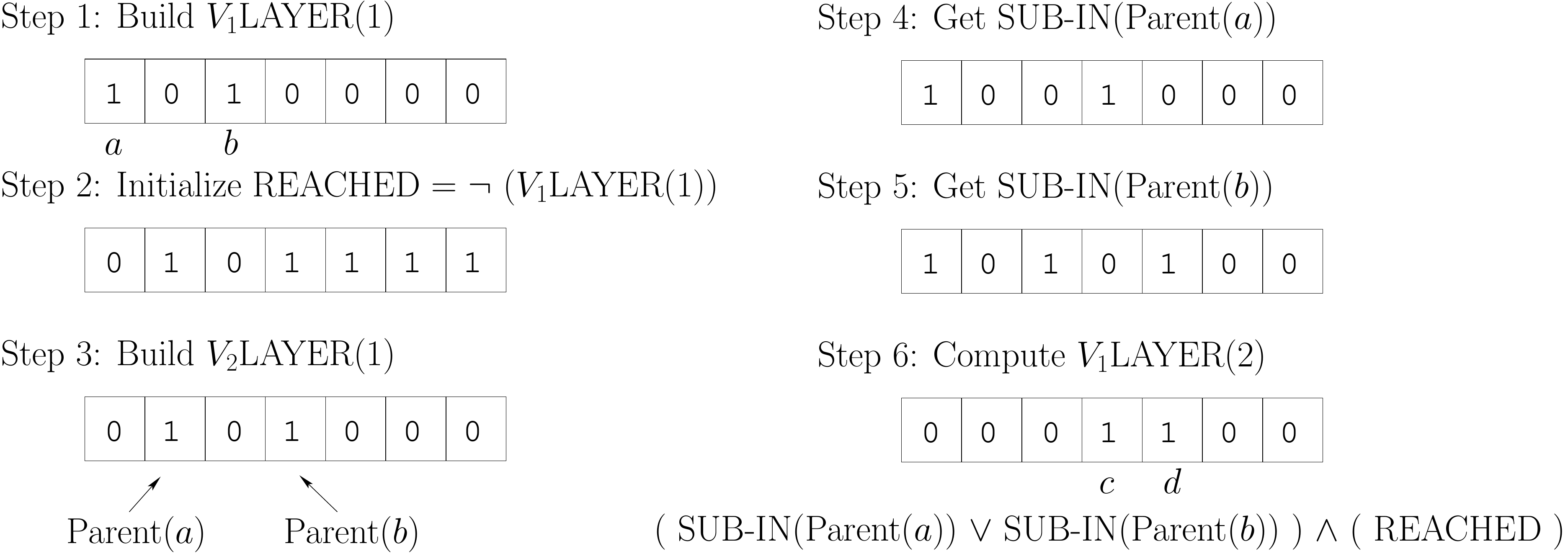}}
\caption{\label{fig:words2}Finding $V_1$-nodes in the 2-layer from $V_2$LAYER(1) using word operations.}
\end{figure}

The $V_2$-nodes in the 2-layer are obtained from $V_1$LAYER(2) (the parent of a $V_1$-node in the 2-layer is a $V_2$-node in the 2-layer). The REACHED word is updated, the $V_2$LAYER(2) word is constructed a bit at a time using the $V_2$-nodes in the 2-layer. As above, $V_1$LAYER(3) is now constructed from the SUB-IN() words of $V_2$-nodes in the 2-layer and REACHED. This continues until there are no more changes in REACHED.

The complexity of constructing the $V_2$LAYER() from the $V_1$LAYER() takes $O(\kappa^2/\lambda)$ throughout the stage. The number of word operations that result in finding a 1-bit in the $V_1$LAYER, and changing the corresponding bit in the $V_2$LAYER is at most $\kappa$, since there are at most $\kappa$ $WT_2$ branches. The number of word operations that result in not finding a 1-bit is $O(\kappa/\lambda)$ for each layer since each $V_1$LAYER() is of length $\kappa$ and we look at the next $\lambda$-word when we do not find a 1-bit. There are at most $\kappa$ layers, so the work done in finding the $V_2$LAYER() words given the $V_1$LAYER() words is $O(\kappa^2/\lambda)$ per stage.

An operation is performed on each $\lambda$-word of SUB-IN() at most once for each node in a stage, and SUB-IN() is of length at most $\kappa$; so the work to generate the layered network is $O(\kappa^2/\lambda$). The work to update the REACHED word is $O(\kappa)$ per stage.

At this point, we have the two-edge distances of all $WT_1$ branches.
To find the set of $ST_1$ immediately reachable from this set, we use
IN() (not SUB-IN() since we want to reach nodes outside the set of
$WT_2$ branches) and the incoming arcs in the adjacency list, in
parallel, for each node to check if a $ST_1$ branch is reachable from
this node. Finding the $\ell$-layer is done analogously to finding the
1-layer. For each node $v$ that is the root of a $WT_2$ branch, an
incoming arc in the adjacency list is scanned for a $ST_1$ neighbor. If
no merger is found, a {\tt lead}() operation is performed on $\lambda$
bits of the IN$^1$($v$) to check for an $ST_1$ neighbor. If no merger
is found, the next arc in the adjacency list is looked at. This
procedure of looking at the next arc in the incoming arcs in the
adjacency list and performing a {\tt lead}() on the next $\lambda$-word
of IN($v$) in parallel continues until a $ST_1$ neighbor is found, or
all the neighbors are exhausted. Since IN($v$) is a word of length
$n_1$, the end of this word is reached in $O(n_1/\lambda)$ operations.
The end of the adjacency list is reached in at most $\kappa$ arc scans
of the adjacency list. Further, each arc is looked at most once so the
total work done throughout the algorithm in checking for $ST_1$
neighbors is $O(\min\{\kappa n_1/\lambda, \kappa^2, m\})$.

We also maintain a word VISITED (of length at most $\kappa$) that keeps track of the branches that have been visited at each stage, i.e., the $i^{th}$ bit of this word is 1 if the root of the branch has {not} been visited in that stage. All bits in this word are initially set to 1.

To push flow through the network, we use $V_2$LAYER(), VISITED, and SUB-OUT() to identify a merger arc. For a node $v \in V_1$ of label $p$, the set of arcs from node $v$ to an unvisited node of layer $p-1$ is found by $V_2$LAYER($p-1$) $\wedge$ SUB-OUT($v$) $\wedge$ VISITED. A {\tt lead}() operation on this resultant word gives a merger arc if it exists.

Each time a merger is found, one more branch becomes visited. Therefore, there can be at most $\kappa$ mergers in each stage. Hence, there are at most $\kappa$ word operations that lead to mergers, which takes $O(\kappa)$ work. Each time a node is revisited in a stage, the search for mergers starts from the last $\lambda$-word checked for a merger; so the work done in word operations that do not find a mergers is $O(\kappa/\lambda)$ per node per stage, which in $O(\kappa^2/\lambda)$ total work per stage.  Updating VISITED requires $O(\kappa)$ work throughout the stage. Hence, work to push flow by executing mergers is $O(\kappa^2/\lambda)$ per stage.

Since each stage can have at most $\kappa$ successful mergers and there are $O(\sqrt{\kappa})$ stages, the number of successful mergers is $O(\kappa^{3/2})$. The IN(), OUT(), SUB-IN() and SUB-OUT() words need to be updated each time a successful merger occurs. Each update takes $O(1)$, so the total work updating these words is $O(\kappa^{3/2})$.

Table \ref{Table:complexity3} summarizes our complexity results for our algorithms with word operations.

\begin{table}[ht]
\begin{center}
\begin{tabular}{||l|c|c||} \hline
\hline Operation & Per stage & Total\\
\hline
Initialization & - & $O(\min\{n_1\kappa, \frac{n_1n_2}{\lambda},m\})$\\
Constructing $1$-layer & - & $O(\min\{\kappa^2, \frac{n_2\kappa}{\lambda},m\})$\\
Constructing $\ell$-layer & -& $O(\min\{\kappa^2,\frac{n_1\kappa}{\lambda}, m\})$ \\
Layered network - layers $2,\ldots ,\ell -1$ & $O(\kappa^{2}/\lambda)$ & $O(\kappa^{2.5}/\lambda)$\\
Executing mergers & $O(\kappa^{2}/\lambda)$ & $O(\kappa^{2.5}/\lambda)$\\
Creating SUB-IN and SUB-OUT & - & $O(\kappa^2)$ \\
Updating SUB-IN, SUB-OUT, IN, and OUT & $O(\kappa)$ & $O(\kappa^{3/2})$ \\
\hline
TOTAL &
\multicolumn{2}{|c||}{$O(\min\{n_1\kappa,\frac{n_1n_2}{\lambda},m\} + \kappa^2 + \kappa^{2.5}/\lambda)$}\\ \hline
\hline
\end{tabular}
\caption{\label{Table:complexity3}Complexity summary of algorithm for bipartite matching with word operations.}
\end{center}
\end{table}

\section{An alternative approach}
\label{sec:alternative}

We now show that it is possible to achieve the theoretical complexity of the {\sf matching-pseudoflow} algorithm by a clever analysis of Hopcroft and Karp's matching algorithm \cite{HopK73}\footnote{We thank an anonymous referee for this analysis}. Given  graph $G = (V_1 \cup V_2, E)$, let the cardinality of the greedy matching be $\kappa_g$. Denote the nodes in the maximal matching by $V_g \subseteq V$, then $|V_g| = 2 \kappa_g$.

\begin{lemma}
\label{lem:altGraphSize}
$\kappa \leq 2 \kappa_g$.
\end{lemma}
\begin{proof}
Every edge in the graph has at least one end point in $V_g$ (otherwise, an edge with neither end point in $V_g$ can be added to the matching, which contradicts maximality). Therefore, every edge in an optimal matching must also have at least one end point in $V_g$.  Thus, the cardinality of the maximum matching is bounded by the cardinality of the set $V_g$, which is $2 \kappa_g$. \qed
\end{proof}

For each $v \in V_g$, let $E_g(v)$ denote the set of edges that have one end point in $v$ and the other end point in another node in $V_g$.  We now construct a graph $G^* = (V_1 \cup V_2, E^*)$, where $E^* \subseteq E$ contains the following edges:
\begin{enumerate}
\item [(i)] For every node $v \in V_g$ with degree $\leq 2 \kappa_g$ in $G$, $E^*$ contains all edges adjacent to $v$.
\item [(ii)]  For every node $v \in V_g$ with degree $ > 2 \kappa_g$ in $G$, $E^*$ contains all edges in $E_g(v)$ and an arbitrary subset of $2 \kappa_g - |E_g(v)|$ edges adjacent to $v$ that are not in $E_g(v)$. That is, a subset of $2 \kappa_g$ edges adjacent to $v$ that contain all the edges in $E_g(v)$.
\end{enumerate}

\noindent Each node $v \in V_g$ in $G^*$ has at most $2\kappa_g$ neighbors by construction. Since every edge is adjacent to some node in $V_g$ and $|V_g| = 2 \kappa_g$, the total number of edges in $E^*$ is at most $4 \kappa_g^2$. Since $\kappa_g \leq \kappa$, $E^*$ has $O(\min\{m, \kappa^2\})$ edges.

\begin{theorem}
\label{theorem:altGraphMatching}
A maximum matching in $G^*$ has cardinality $\kappa$.
\end{theorem}
\begin{proof}

\begin{figure}[ht]
\centerline{\includegraphics[width=0.5\linewidth]{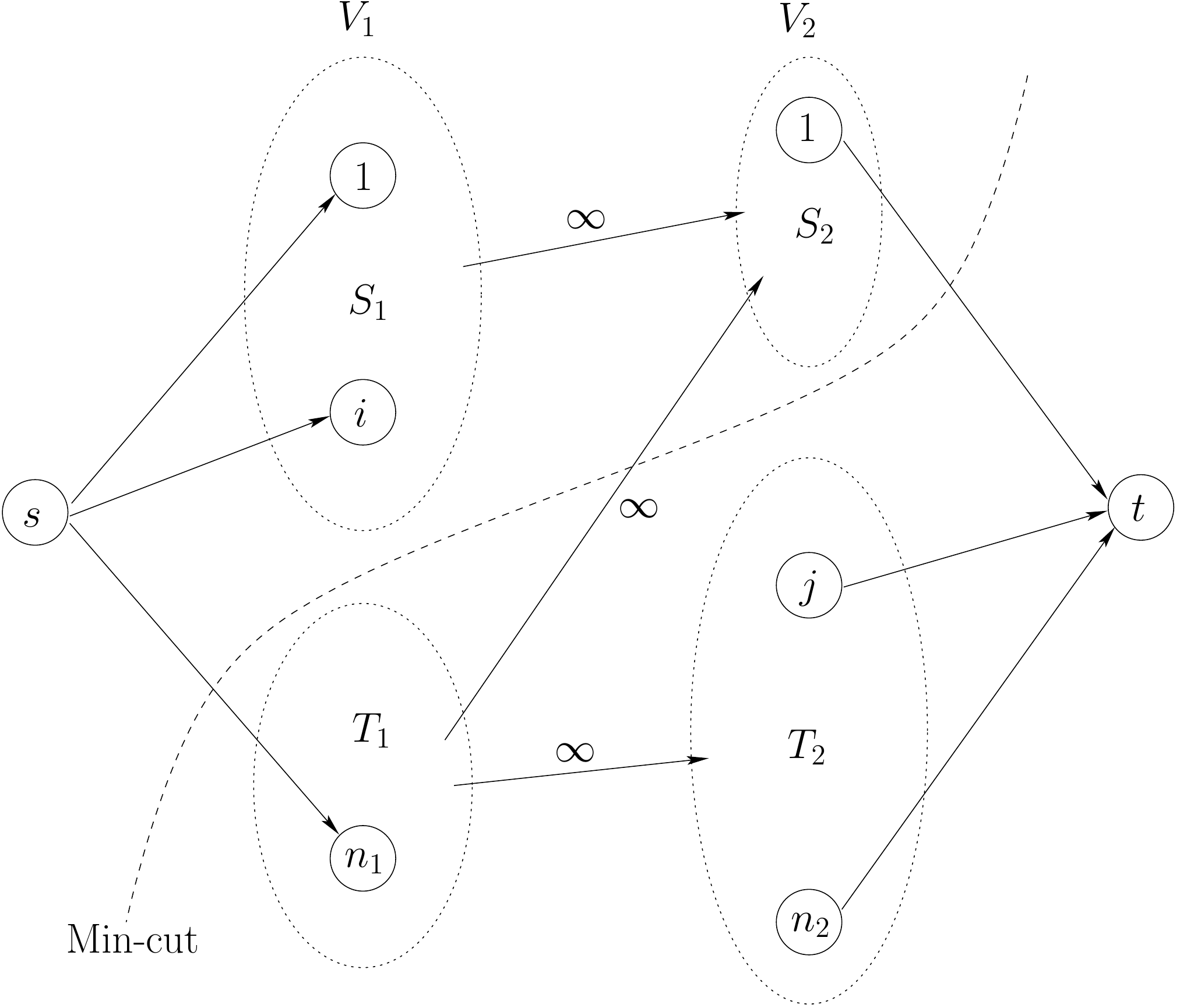}}
\caption{\label{fig:mincut}Minimum cut in $G^*_{st}$.}
\end{figure}

Let the cardinality of a maximum matching in $G^*$ be denoted by $\kappa^*$. We obtain the maximum matching by solving for a minimum cut in a graph $G^*_{st}$ obtained by adding a source node $s$, a sink node $t$, and unit capacity arcs from $s$ to all nodes in $V_1$ and from all nodes in $V_2$ to $t$. Let the source set of the minimum cut in be $S^* = \{s\} \cup S_1 \cup S_2$ and the sink set be $T^* = \{t\} \cup T_1 \cup T_2$ where $S_1 \cup T_1 = V_1$ and $S_2 \cup T_2 = V_2$ as shown in Figure \ref{fig:mincut}. Then, the capacity of this minimum cut is $\kappa^*$, i.e., $|T_1| + |S_2| = \kappa^*$.

The maximum matching in $G$ is similarly obtained by solving for a minimum cut in a graph $G_{st}$ (obtained by adding a source and a sink node, and arcs adjacent to the source and sink); this minimum cut has capacity $\kappa \leq n_1 \leq n_2$.

Suppose (for contradiction) that $\kappa^* < \kappa$. Then, $S_1,\ S_2,\ T_1, \mbox{ and } T_2$ are non-empty (if any of these sets were empty, then  $\kappa^* = n_1$ which contradicts the assumption that $\kappa^* < \kappa \leq n_1$). The minimum $s,t$-cut $(S^*, T^*)$ in $G^*_{st}$ cannot be a finite cut in $G_{st}$ since it has a capacity strictly less than the minimum cut in $G_{st}$. Then, there exists some arc $(i,j)$ in $G_{st}$ but not in $G^*_{st}$ such that $i \in S_1$ and $j \in T_2$. Since the arc $(i,j)$ was removed from $G$ to generate $G^*$, it means that node $i$ has exactly $2 \kappa_g$ neighbors in $G^*$.  Further, since $i$ is in the source set of a finite cut in $G^*_{st}$, all the neighbors of $i$ must belong to $S_2$. That is, $|S_2| \geq 2 \kappa_g$. We have shown that $\kappa^* > |S_2|$ since $T_1$ is non-empty. Therefore, $\kappa^* > |S_2| \geq 2 \kappa_g \geq \kappa$, contradicting the assumption that $\kappa^* < \kappa$.  \qed
\end{proof}

The above observations and theorem suggest the following algorithm:
\begin{enumerate}
\item Generate a maximal matching (takes $O(\min\{m, n_1 \kappa\})$ work).
\item Construct graph $G^*$ as described above (takes $O(\min\{m, \kappa^2\})$ work).
\item Solve for a maximum matching using the Hopcroft-Karp algorithm. Since the number of edges in $G^*$ is $O(\min\{m, \kappa^2\})$, and the number of nodes is $O(\kappa)$, the complexity is $O(\sqrt{\kappa} \min\{m, \kappa^2\})$.
\end{enumerate}

\section{Test instances}
\label{sec:instances}

The descriptions of these instances is reproduced from Cherkassky et al. \cite{CheGMSS98}.

\begin{enumerate}
\item {\bf Fewg and manyg}: These are random bipartite graphs where the vertices of each partition, $V_1$ and $V_2$ , are divided into $k$ groups of equal size. For each vertex of the $j$-th group of $V_1$ the generator chooses $y$ random neighbors from the $(i-1)$-th through $(i+1)$-th groups of $V_2$ (with wrap-around), where $y$ is binomially distributed with mean $d$ (thus $d$ = mean vertex degree). The indices $i$ and $j$ are not related because vertices in $V_1$ are randomly shuffled before neighbors in $V_2$ are assigned. The two families we consider are {\sf fewg}, where there are 32 groups, and {\sf manyg}, where there are 256 groups; both have $d$ = 5.

These classes were designed having in mind problems that can be reduced to bipartite matching, such as the maximum vertex-disjoint paths problem. In these problems the resulting graph in the reduction is bipartite, but if the original graph is planar or nearly planar each vertex will only have as neighbors vertices in the surrounding area.

\item {\bf Hilo}: The {\sf hi-lo} family of bipartite matching problems was designed to separate high and low vertex selection strategies for the push-relabel method. This generator creates a graph with a unique perfect matching and has been motivated by a generator of Kennedy \cite{Ken95}.

Let $G = (V_1;V_2, E)$ be a graph produced by this generator. This graph is defined by three parameters, $\ell$, $k$, and $d$. Vertices of $V_1$ are partitioned into $\ell$ groups, each containing $k$ vertices. For $1 \leq i \leq k$, $1 \leq j \leq \ell$, we refer to the $i$-th vertex in group $j$ by $x^j_i$. Vertices of $V_2$ are partitioned similarly, and $y^j_i$ is defined similarly to $x^j_i$. Each vertex $x^j_i$ is connected to vertices $y^j_p$ for $\max(1, i-d) \leq p \leq i$ and, if $j < \ell$, to vertices $y^{j+1}_p$ for $\max(1, i-d) \leq p \leq i$.

\item {\bf Grid}: In class {\sf grid}, each vertex $u \in V_1$ is connected to vertices $\{u+1,\ u-1,\ u+a,\ u-a,\ u+b,\ u-b, \ldots\}$ where $\{1,\ a,\ b, \ldots\}$ is a geometric progression. In our tests, we set the average degree of each node to 6.

\item {\bf Hexa}: In class {\sf hexa}, the vertices on each side are divided into $n/b$ blocks of size $b$. One random bipartite hexagon is added between each block $i$ on one side and each of the blocks $i + k$ on the other side, with $|k| \leq K$ for some $K$. The parameters $b$ and $K$ are chosen by the program in such a way that the average degree is correct (i.e., $3K/b = d$) but few pairs of hexagons have more than one vertex in common. In our tests, we set $d$ = 6.

\item {\bf Rope}: For the class {\sf rope}, the vertices on each side are grouped into $t = n/d$ blocks of size $d$, numbered $V_1^0 \ldots V_1^{t-1}$ and $V_2^0 \ldots V_2^{t-1}$. Block $i$ on one side is connected to block $i+1$ on the other side, for $i = 0, 1, \ldots, t-2$; block $V_1^{t-1}$ is connected to block $V_2^{t-1}$. Thus, the graph is a "rope" that is folded and twisted over itself, so that it zig-zags between the two sides, first up and then down. Consecutive
pairs of blocks along the ``rope'' are connected alternately by perfect matchings (``$m$-type arcs'') and random bipartite graphs of average degree $d-1$ (``$r$-type arcs''), beginning and ending with perfect matchings. The only maximum matching is a perfect one, consisting of all $m$-type arcs. In our tests, set $d=6$.

\item {\bf Zipf}: Each member of class {\sf zipf} is a random bipartite graph where the arc between the $i$-th $V_1$-node and the $j$-th $V_2$-node has nominal probability roughly proportional to $1/(ij)$. Thus the graph is denser near the ``core'' vertices (those with small index), and thins out slowly towards the ``periphery'' (vertices with high index). In our experiments we set $d = 6$.
\end{enumerate}

\end{document}